\documentclass[12pt]{article}

\usepackage{booktabs}

\usepackage[T1, OT1]{fontenc}
\DeclareTextSymbolDefault{\DH}{T1}
\usepackage{amssymb,latexsym}
\usepackage{graphicx,amsmath,amsfonts}
\usepackage{comment}
\usepackage{tikz}
\usetikzlibrary{arrows}
\usetikzlibrary{calc}
\usepackage{pgfplots}
\usepackage{chngcntr}
\usepackage{apptools}
\AtAppendix{\counterwithin{lemma}{section}}
\AtAppendix{\counterwithin{theorem}{section}}
\AtAppendix{\counterwithin{proposition}{section}}
\AtAppendix{\counterwithin{example}{section}}

\usepackage{url}
\usepackage{mathrsfs}
\usepackage{colortbl}
\definecolor{LightCyan}{rgb}{0.88,1,1}
\definecolor{celadon}{rgb}{0.67, 0.88, 0.69}
\definecolor{columbiablue}{rgb}{0.61, 0.87, 1.0}
\usepackage{paralist}
\usepackage{hyperref}
\usepackage[sort&compress,nameinlink]{cleveref}
\usepackage{nicefrac}
\usepackage{booktabs}
\usepackage{setspace}
\usepackage{etoolbox}
\usepackage{wrapfig}
\usepackage[round]{natbib}
\usepackage{amsthm}
\usepackage{makecell}
\allowdisplaybreaks
\usepackage[nottoc,numbib]{tocbibind}
\usepackage{enumerate}
\usepackage{subcaption}
\usepackage{bbold}

\usepackage[leftbars]{changebar}
\makeatletter

\hypersetup{
colorlinks=true,citecolor=blue!80!black,linkcolor=red!60!black,urlcolor=blue!80!black,
pagebackref=true
}

\usepackage{pifont}

\newtheorem{claim}{Claim}

\usepackage{etoolbox}
\newcommand{\addQEDstyle}[2]{\AtBeginEnvironment{#1}{\pushQED{\qed}\renewcommand{\qedsymbol}{#2}}
\AtEndEnvironment{#1}{\popQED}}
\addQEDstyle{example}{$\dashv$}

\newcommand{\calR}{\mathcal{R}}

\makeatletter
\newtheorem*{rep@theorem}{\rep@title}
\newcommand{\newreptheorem}[2]{%
\newenvironment{rep#1}[1]{%
 \def\rep@title{#2 \ref{##1}}%
 \begin{rep@theorem}}%
 {\end{rep@theorem}}}
\makeatother
\newreptheorem{theorem}{Theorem}
\newreptheorem{proposition}{Proposition}
\newreptheorem{lemma}{Lemma}
\newreptheorem{corollary}{Corollary}

\definecolor{OKgreen}{HTML}{035925}
\definecolor{NOred}{HTML}{8F061F}

 \definecolor{OKgreen}{HTML}{6B9B00}
\definecolor{NOred}{HTML}{83004F}
\definecolor{BGblue}{HTML}{4575D4}
\definecolor{BGgrey}{HTML}{808080}
\usepackage{pifont}%
\newcommand*\rot{\rotatebox{75}}
\newcommand*\OK{\textcolor{OKgreen}{\ding{51}}}
\newcommand*\NO{\textcolor{NOred}{$\times$}}

\newcommand*\Next{\textit{Next-$k$}}
\newcommand*\First{\textit{First Majority}}
\newcommand*\Approval{\textit{Approval Voting}}
\newcommand*\Threshold{\textit{$f$-Threshold}}
\newcommand*\Majority{\textit{$0.5n$-Threshold}}
\newcommand*\Fgap{\textit{First $k$-Gap}}
\newcommand*\TopFgapSK{\textit{Top-$s$-First-$k$-Gap}}
\newcommand{\TopsFgapSK}[2]{\textit{Top-$#1$-First-$#2$-Gap}}
\newcommand*\TopFgap[2]{\textit{Top-$#1$-First-$#2$-Gap}}
\newcommand*\Ffgap{\textit{First $5$-Gap}}
\newcommand*\kAV{\textit{Multi-winner Approval Voting}}
\newcommand*\MFgap{\textit{Modified First $\ell$-Gap}}
\newcommand*\Lgap{\textit{Largest Gap}}
\newcommand*\SPAV{\textit{Size Priority}}
\newcommand*\ISPAV{\textit{Increasing Size Priority}}
\newcommand{\ispav}[1]{\textit{ISP-#1}}
\newcommand*\DSPAV{\textit{Decisive Size Priority}}
\newcommand*\qNCSA{\textit{q-NCSA}}
\newcommand*\MSThreshold{\textit{Max-Score-$f$-Threshold}}

\newtheorem{Thm}{Theorem}
\newtheorem{Prop}[Thm]{Proposition}

\theoremstyle{definition}
\newtheorem{Def}{Definition}

\newtheorem{Exp}{Example}
\newtheorem*{Exp*}{Example}

\newtheorem{Obs}{Observation}

\newtheorem{Axm}{Axiom}
\newtheorem{Rule}{Rule}
\newtheorem*{Rule*}{Rule}

\newcommand{\cR}{\mathcal{R}}
\newcommand{\score}{\mathit{sc}}
\newcommand{\fscore}{\mathit{fsc}}

\begin{document}

\title{Approval-Based Shortlisting}

\author{Martin Lackner\\
  TU Wien\\
  Vienna, Austria\\
{\small \texttt{lackner@dbai.tuwien.ac.at}}
\and
Jan Maly\\
 ILLC, University of Amsterdam\\
 Amsterdam, The Netherlands\\
{\small \texttt{j.f.maly@uva.nl}}}

\maketitle 

\begin{abstract}
Shortlisting is the task of reducing a long list of alternatives 
to a (smaller) set of best or most suitable alternatives.
Shortlisting is often used in the nomination process of awards
or in recommender systems to display featured objects.
In this paper, we analyze shortlisting methods that are based on approval data,
a common type of preferences.
Furthermore, we assume that the size of the shortlist, i.e., the number 
of best or most suitable alternatives, is not fixed but determined 
by the shortlisting method. %
We axiomatically analyze established and new shortlisting methods
and complement this analysis with an experimental evaluation based on synthetic and real-world data.
Our results lead to recommendations which shortlisting methods to use, depending on the desired properties.

\end{abstract}

\section{Introduction}

Shortlisting is a task that arises in many different scenarios and applications:
given a large set of alternatives, identify a smaller subset that consists of
the best or most suitable alternatives. 
Prototypical examples of shortlisting are awards where a winner must be selected among a vast
number of eligible candidates. In these cases, we often find a two-stage process.
In a first shortlisting step, the large number of contestants (books, films, individuals, etc.)
is reduced to a smaller number. In a second step, the remaining contestants 
can be evaluated more closely and one contestant in the smaller set is chosen to receive the award.

Both steps may involve a form of group decision making (voting),
but can also consist of a one-person or even automatic decision.
For example, the shortlist of the Booker Prize is selected by a small jury~\citep{Booker},
whereas the shortlists of the Hugo Awards are compiled based on thousands of ballots~\citep{Hugo}.
Similarly, the Baseball Writers' Association of America selects the new entries
into the Baseball Hall of Fame by voting. In that case, any 
candidate with at least 75\% approval enters the hall of fame, without a second round.
Another very common application of shortlisting is the selection
of most the promising applicants for a position who will be invited for an
interview~\citep{bovens2016selection,singh2010prospect}.
Apart from these prototypical examples, shortlisting is also useful in many less obvious applications
like the aggregation of expert opinions for example in the medical domain~\citep{clemens} or in 
risk management and assessment~\citep{tweeddale1992some}. %
Shortlisting can even be used in scenarios without agents in the traditional sense, for example
if we consider features as voters to perform an initial screening of objects,
i.e., a feature approves all objects that 
exhibit this feature~\citep{Faliszewski17OA}.

In this paper, we consider shortlisting as a form of collective decision making. %
We assume that a group of voters announce their preferences by specifying which alternatives
they individually view worthy of being shortlisted, i.e., they file approval ballots.
In practice, approval ballots are commonly used for shortlisting,
because the high number of alternatives that necessitates shortlisting
in the first place precludes the use of ranked ballots.
Furthermore, we assume that the number of alternatives to be shortlisted is not fixed
(but there might be a preferred number),
as there are very few shortlisting scenarios where there is a strong justification
for an exact size of the shortlist. 
Due to this assumption, we are not in the classical setting
of multi-winner voting~\citep{kilgour2012approval,Faliszewski2017MultiwinnerVoting,lackner2020approval},
where a fixed-size committee is selected, but in the more general setting of multi-winner voting
with a variable number of winners~\citep{kilgour2010approval,Kilgour16,Faliszewski17OA}.

In real-world shortlisting tasks, there are two prevalent methods in use: \kAV{} (selecting the $k$~alternatives with the highest approval score) and threshold rules
(selecting all alternatives approved by more than a fixed percentage of voters).
Further shortlisting methods have been proposed in the literature~\citep{Brams2012,Kilgour16,Faliszewski17OA}.
Despite the prevalence of shortlisting applications, there does not exist work on systematically choosing a suitable shortlisting method.
Such a recommendation would have to consider both expected (average-case) behavior and guaranteed axiomatic properties, and
neither have been studied previously specifically for shortlisting applications (cf.\ related work below).
Our goal is to answer this need and provide principled recommendations for shortlisting rules, depending on the properties that are desirable in the specific shortlisting process.

\smallskip
\noindent In more detail, the contributions of this paper are the following:
\begin{itemize}
\item We define shortlisting as a voting scenario and specify minimal requirements for shortlisting methods (Section~\ref{sec:formalmodel}). Furthermore, we introduce five new shortlisting methods: \Fgap{}, \Lgap, \TopFgapSK, \MSThreshold, and \SPAV{} (Section~\ref{sec:rules}).
\item We conduct an axiomatic analysis of shortlisting methods and by that identify essential differences between them. Furthermore, we axiomatically characterize \Approval, \Threshold, and the new \Fgap{} rule (Section~\ref{sec:axiomatic}).
\item We present a connection between shortlisting and clustering algorithms, as used in machine learning. We show that \Fgap{} and \Lgap{} can be viewed as instantiations of linkage-based clustering algorithms (Section~\ref{sec:clustering}).%
\item In numerical simulations using synthetic data, we approach two essential difficulties of shortlisting processes: we analyze the effect of voters with \emph{imperfect (noisy) perception} of the alternatives and the effect of \emph{biased voters}. These simulations complement our axiomatic analysis by highlighting further properties of shortlisting methods and provide additional data points for recommending shortlisting methods (Section~\ref{sec:experiments}).

\item In addition to synthetic data, we collected voting data from the Hugo Awards, which are annual awards for works in science-fiction. This data set is a real-world application of shortlisting and offers a challenging test-bed for shortlisting rules.
Using this data set, we investigate the ability of different shortlisting rules to produce short shortlists without excluding the alternative that actually won the award (Section~\ref{sec:experiments}).

\item
An open-source implementation~\citep{martin_lackner_2020_3821983} of all considered shortlisting rules and the numerical experiments is available, including the Hugo data set.

\item The recommendations based on our findings are summarized in Section~\ref{sec:concl}. In brief, our analysis leads to a recommendation of \TopFgapSK, \Threshold{}, and \SPAV{}, depending on the general shortlisting goal and desired behavior.
\end{itemize}

A preliminary version of this work has appeared in the proceedings of the 20th International Conference on Autonomous Agents and Multiagent Systems (AAMAS 2021) \citep{aamas-shortlisting}.

\subsection*{Related Work}
There are two recent papers that are particularly relevant for our work. Both \citet{Faliszewski17OA} and  \citet{prop-mwwavnow} investigate multi-winner voting with a variable number of winners. 
In contrast to our paper, the main focus of \citep{Faliszewski17OA} lies on computational complexity, which is less of a concern for our shortlisting setting (as discussed later).
The paper also contains numerical simulations related to the number of winners (which is one of the metrics we consider in our paper).
In the few cases where shortlisting rules are considered\footnote{
\citet{Faliszewski17OA} consider 
\textit{NAV} and \textit{NCSA}, which are
equivalent to \Threshold{} and \Approval{} in our paper (subject to tiebreaking), as well as 
\First{} and \qNCSA.}, their results regarding the average size of winner sets agree with our simulations (Section~\ref{sec:experiments}). 

\citet{prop-mwwavnow} study proportionality in shortlisting scenarios. 
A proportional representation of voters is incompatible with our desiderata of shortlisting rules (i.e., proportionality is incompatible with the Efficiency axiom, which we require for shortlisting rules). Thus, the rules and properties considered in~\citep{prop-mwwavnow} do not intersect with ours and are difficult to compare with.
A simplified separation between our work and theirs is the underlying assumption of fairness:
we require that the most deserving candidates are included in the shortlist (fairness towards
candidates), whereas proportionality is concerned with fairness towards voters.

There are two other notable voting frameworks with a variable number of winners.
First, shortlisting rules can be viewed as a particular type of \emph{social dichotomy functions}~\citep{duddy2014social,brandl2019axiomatic}, i.e., voting rules which partition alternatives into two groups.
Moreover, multiwinner voting with a variable number of winners can be seen as 
a special case of (binary) \emph{Judgment Aggregation}~\citep{list2012JA,Handbook-JA} 
without consistency constraints. However, both of these frameworks treat the set of selected 
winners and its complement as symmetric. This is in contrast to shortlisting,  where
we usually expect the winner set to be only a small minority of all available candidates.
For this reason, social dichotomy functions and Judgment Aggregation rules are generally
not well suited for shortlisting.

It is worth mentioning that shortlisting is
is not only studied as a form of 
collective decision making but 
also studied as a model of individual decision making.
\citet{ManziniMariotti} proposed \emph{Rational Shortlisting Methods}
as a model of human choice, which lead to number of works on shortlisting
as a decision procedure, for example \citep{dutta2015inferring}, \citep{horan2016simple},
\citep{kops2018f}, and \citep{tyson2013behavioral}.

More generally, there is a substantial literature on multi-winner voting with a \emph{fixed} number of winners (i.e., committee size), as witnessed by recent surveys~\citep{kilgour2012approval,Faliszewski2017MultiwinnerVoting,
lackner2020approval}. Multi-winner voting rules are much better understood, both from an axiomatic~\citep{elkind2017properties,Fernandez2017Proportional,
aziz2017justified,jet-consistentabc,sanchez2019monotonicity} and experimental~\citep{elkind2017multiwinner,bredereck2019experimental} point of view, also in the context of shortlisting~\citep{ijcai/mw-condorcet,bredereck2018coalitional}. Results for multi-winner rules, however, typically do not easily translate to the setting with a variable number of winners.
 
\section{The Formal Model}\label{sec:formalmodel}

In this section we describe our formal model that embeds shortlisting in a voting framework.
The model consists of two parts: a general framework for
approval-based elections with a variable number of winners~\citep{kilgour2010approval,Kilgour16,Faliszewski17OA}
on the one hand and, on the other hand, four basic axioms
that we consider essential prerequisites for shortlisting rules.

An approval-based election $E =(C,V)$ consists of a non-empty set of candidates (or alternatives)\footnote{In the following, we use the words candidate and candidate interchangeably.} $C = \{c_1,\dots, c_m\}$
and an $n$-tuple of approval ballots $V =(v_1,\dots,v_n)$ where $v_i \subseteq C$.
If $c_j \in v_i$, we say that voter~$i$ approves candidate $c_j$; if $c_j \not \in v_i$, voter~$i$ does not approve candidate $c_j$.
We interpret a voter's approval of a candidate as the preference for this candidate being included in the shortlist.
In the following we will always write $n_E$ for the number of voters and $m_E$ for 
the number candidates in an election $E$.
We will omit the subscript if $E$ is clear from the context.

The \emph{approval score} $\score_E(c_j)$ of candidate $c_j$
in election $E$ is the number of approvals of $c_j$ in $V$, i.e., $\score_E(c_j)=\lvert \{i : 1\leq i \leq n \text{ and } c_j\in v_i\}\rvert$.
We write $\score(E)$ for the vector $(\score_E(c_1), \dots$, $\score_E(c_m))$.
To avoid unnecessary case distinctions,
we only consider \emph{non-degenerate} elections:
these are elections where not all candidates have the same approval score.

An \emph{approval-based variable multi-winner rule} (which we refer to just as ``voting rule'')
is a function mapping an election $E=(C,V)$ to a subset of $C$.
Given a rule $\cR$ and an election $E$, $\cR(E)\subseteq C$ is the \emph{winner set} according to voting rule
$\cR$, i.e., $\cR(E)$ is the set of candidates which have been shortlisted.
Note that $\cR(E)$ may be empty or contain all candidates.
We refer to candidates in $\cR(E)$ as \emph{winners} or \emph{winning candidates}.

Now we introduce the basic axioms that we require every shortlisting rule to satisfy.
Anonymity and Neutrality are two basic fairness axioms for voting rules \citep{Handbook-Voting}.

\begin{Axm}[Anonymity] All voters are treated equal,
i.e., for every permutation $\pi: \{1,\dots,n\}\to \{1,\dots,n\}$ and election $E=(C,V)$ where 
$V= (v_1, \dots, v_n)$, if $E^*=(C,V^*)$ with $V^* = (v_{\pi(1)},\dots,v_{\pi(n)})$,
then $\cR(E) = \cR(E^*)$.
\end{Axm}

\begin{Axm}[Neutrality] All candidates are treated equally,
i.e., for every election $E =(C,V)$ where $V= (v_1, \dots, v_n)$
and permutation $\pi: C\to C$, if $E^* = (C, V^*)$ where $V^* =(v_1^*,\dots,v_n^*)$
with $v_i^* = \{\pi(c) \mid c \in v_i\}$,
then $\pi(c) \in \cR(E^*)$ iff $c \in \cR(E)$ for all $c\in C$.
\end{Axm}

Shortlisting differs from other multi-winner scenarios in that
we are not interested in representative or proportional 
committees. Instead, the goal is to select the most excellent
candidates. This goal is formalized in the following axiom.

\begin{Axm}[Efficiency]
No winner set can have a strictly smaller approval score than a non-winner, i.e.,
for all elections $E=(C,V)$ and all candidates $c_i, c_j \in C$ 
if $\score_E(c_i) > \score_E(c_j)$ and $c_j \in \cR(E)$
then also $c_i \in \cR(E)$.
\end{Axm}

The assumption that approval scores are approximate measures of the general quality of candidates can also be argued in a probabilistic framework: under reasonable assumptions a set of candidates with the highest approval scores coincides with the maximum likelihood estimate of the truly best candidates~\citep{procaccia2015approval}. Thus, we impose Efficiency to guarantee the inclusion of the most-likely best candidates.

Efficiency can also be argued for from the perspective of voters:
Let $\calR$ satisfy Efficiency and $W=\calR(E)$ for some election $E$. Then we claim that there does not exist a set $W'$ with $\lvert W\rvert=\lvert W'\rvert$ such that (i) $\lvert W'\cap v\rvert\geq \lvert W\cap v\rvert$ for all $v\in V$ and (ii) for some $w\in V$ $\lvert W'\cap w\rvert> \lvert W\cap w\rvert$. Otherwise $\sum_{c\in W'}\score_E(c) > \sum_{c\in W}\score_E(c)$ would hold. As $\lvert W\rvert = \lvert W'\rvert$ this implies that there is at least one candidate $c \in W' \setminus W$ with $\score_E(c) > \min\{\score_E{c'} \mid c' \in W\}$, a contradiction. In this sense, efficient shortlists are \emph{Pareto efficient among shortlists of the same size}.

It is also worth noting that Efficiency rules out proportional voting rules. It is easy to see why: a proportional selection of winner sets has to contain candidates supported by (sufficiently sized) minorities. As Efficiency demands that majority candidates are always to be preferred, any sensible notion of proportionality clashes with Efficiency.

The last of our basic axioms is Non-tiebreaking.
Since the number of winners is variable in our setting, there is generally no 
need to break ties. Because tiebreaking is usually an arbitrary  
and unfair process,
voting rules should not introduce unnecessary tiebreaking.
This idea yields our fourth axiom:

\begin{Axm}[Non-tiebreaking]
If two candidates have the same approval score, either both or neither should be winners.
That is, for all elections $E=(C,V)$ and all candidates $c_i$ and $c_j$
if $\score_E(c_i) = \score_E(c_j)$ then
either $c_i,c_j \in \cR(E)$ or $c_i,c_j \not \in \cR(E)$.
\end{Axm}

We postulate these four axioms as the minimal requirements for a voting rule to be 
considered a shortlisting rule in our sense.

\begin{Def}\label{def:shortlisting}
An approval-based variable multi-winner rule is
a shortlisting rule if it satisfies Anonymity, Neutrality, Efficiency and is non-tiebreaking.
\end{Def}

Observe that Non-tiebreaking and Efficiency are axioms that are only interesting 
if we consider voting with a variable number of winners.
Clearly, no voting rule for voting with a fixed number of winners can be non-tiebreaking.
Furthermore, except for the issue of how to break ties, there is exactly one voting rule
for approval voting with a fixed number $k$ of winners that satisfies Efficiency,
namely picking the $k$ candidates with maximum approval score (\kAV).

A consequence of Efficiency and Non-tiebreaking is that a shortlisting rule only has
to decide how many winners there should be.
This reduces the complexity of finding the winner set drastically 
as there are only linearly many possible winner sets, in contrast to the exponentially many subsets of $C$.

\begin{Obs}\label{Obs1}
For every election there are at most $m +1$ sets that can be winner sets under a shortlisting rule.
\end{Obs}

\section{Shortlisting Rules}\label{sec:rules}

In the following, we define the 
shortlisting rules that we study in this paper.
We define these rules by specifying which properties a candidate has to satisfy to be contained in the winner sets.
As before, let $E=(C,V)$ be an election.
We assume additionally that
$c_1, \dots, c_m$ is an enumeration of the candidates
in descending order of approval score, i.e.,
such that $\score_E(c_{i-1}) \geq \score_E(c_i)$ for all $2 \leq i \leq m$.
We will illustrate all rules on the following example:

\begin{Exp}\label{Exp1}
Let $E = (C,V)$ be an election with $10$ voters and $8$ candidates $c_1, \dots, c_8$.
The scores are given by $\score(E) = (10,10,9,8,6,3,3,0)$. This instance is illustrated
in Figure~\ref{fig:Exp1}.
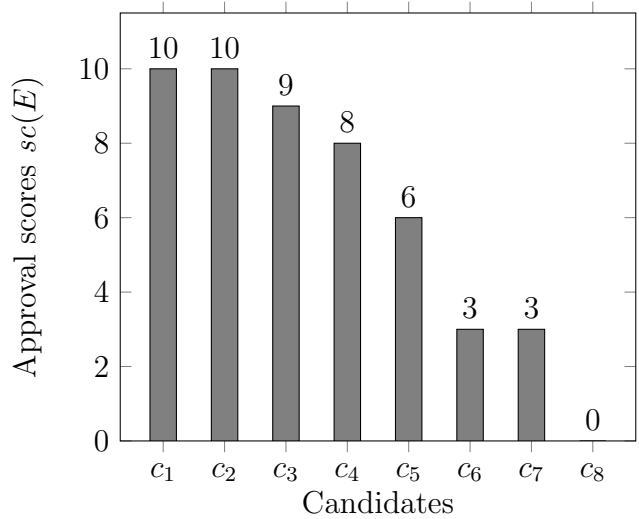
\begin{figure}
\begin{center}
\begin{tikzpicture}
\begin{axis}[
    ybar,
    ymin=0,
    ymax= 11.5,
    ylabel={Approval scores $\score(E)$},
    xlabel={Candidates},
    bar width = 10pt,
    symbolic x coords={$c_1$,$c_2$,$c_3$,$c_4$,$c_5$,$c_6$,$c_7$,$c_8$},
    xtick=data,
    nodes near coords,
    nodes near coords align={vertical},
    ]
\addplot[fill = BGgrey] coordinates {($c_1$,10) ($c_2$,10) ($c_3$,9) ($c_4$,8) ($c_5$,6) ($c_6$,3) ($c_7$,3) ($c_8$,0)};
\end{axis}
\end{tikzpicture}
\end{center}
\caption{Illustration of Example~\ref{Exp1}}
\label{fig:Exp1}
\end{figure}
There are seven possible winner sets for a shortlisting rule: $\emptyset$, $\{c_1,c_2\}$, $\{c_1, c_2, c_3\}$, $\{c_1, c_2, c_3, c_4\}$, $\{c_1, c_2, \dots, c_5\}$, $\{c_1, c_2, \dots, c_7\}$, $\{c_1, c_2, \dots, c_8\}$.
\end{Exp}

\subsection{Established Rules}

First we introduce the shortlisting rules that are either commonly used in practice
or have been proposed in the literature before.
A natural idea
is to select all most-approved candidates.
The corresponding winner set equals the set of co-winners
of classical Approval Voting \citep{Brams78}.

\begin{Rule}[Approval Voting]
A candidate~$c$ is a winner iff $c$'s approval score is maximal, i.e.,
$c \in \cR(E)$ iff $\score_E(c)  = \max(\score(E))$.
\end{Rule}

The winners under \Approval{} in Example~\ref{Exp1} are
$c_1$ and $c_2$ as they both have the highest score.

Another natural way to determine the winner set is 
to fix some percentage threshold and declaring
all alternatives to be winners that surpass this approval threshold~\citep{kilgour2010approval}.
For example, for a baseball player to be entered into the Hall of Fame, more than 75\% of the members of the Baseball Writers' Association of America 
have to approve this nomination~\citep{Baseball}.
Such rules are known as quota rules in judgment aggregation~\citep{Handbook-JA}.

\begin{Rule}[$f$-Threshold]
Let $f:\mathbb{N} \to \mathbb{N}$ be a function
such that $0 < f(\lvert V\rvert) < \lvert V\rvert$.
Then, $c \in \cR(E)$ 
for an alternative $c \in C$
if and only if $\score_E(c) > f(\lvert V\rvert)$.
We write $\alpha$-Threshold for a constant $0\leq \alpha < 1$ to denote the $f$-Threshold rule with $f(n)=\lfloor\alpha\cdot n\rfloor$.
\end{Rule}

Consider for example $f(\lvert V\rvert) = \frac{\lvert V\rvert}{2}$. Then
an alternative is a winner if it is approved by more than 50\% of all voters.
In Example~\ref{Exp1} this would mean that the winner set contains all candidates
with 6 or more approvals, i.e., $c_1, \dots, c_5$.

A sensible modification of \Threshold{} would be to select all alternatives
with an above-average approval score, i.e., the set of winners consists of all alternatives $c$
with $\score_E(c)>\frac{1}{m}\cdot \sum_{c'\in C} \score_E(c')$.
This rule is also a shortlisting rule in our sense. However, as it will, in expectation,
select half of the available candidates, we do not think that it is a reasonable rule in most
shortlisting settings. Therefore, we do not study it and only mention
that it might be a good rule in other voting contexts with a variable number of winners.
For example, Duddy et al.~\citep{duddy2016} analyzed this rule and concluded that it is 
the best rule for partitioning alternatives into homogeneous groups (see also the axiomatic characterization of this rule in~\citep{brandl2019axiomatic}). 

Another natural modification is to base the threshold not on the number of voters
but on the highest approval score achieved by a candidate. We call this 
\MSThreshold. This variant of \Threshold{} turns out to be well suited
to shortlisting as it formalizes the goal of selecting all candidates
that are close to the top.

\begin{Rule}[Max-Score-$f$-Threshold]
Let $f:\mathbb{N} \to \mathbb{N}$ be a function
such that $0 < f(x) < x$.
Then, $c \in \cR(E)$ 
for an alternative $c \in C$
if and only if $\score_E(c) > f(\max \score(E))$.
We write Max-Score$\alpha$-Threshold for a constant $0\leq \alpha < 1$  to denote the Max-Score-$f$-Threshold rule with $f(n)=\lfloor\alpha\cdot n\rfloor$.
\end{Rule}

We observe that $c_1$ and $c_2$ in Example~\ref{Exp1} have score $n$, hence 
\Threshold{} and \MSThreshold{} coincide on the example.

The next three rules are further shortlisting methods that have been proposed in the literature.
\First{} \citep{Kilgour16} includes as many alternatives as necessary to comprise more than half of all approvals.
The following definition deviates slightly from the original definition
\citep{Kilgour16}
in that it is non-tiebreaking.

\begin{Rule}[First Majority]
Let $i$ be the smallest index such that 
$\sum_{j \leq i} \score_E(c_j) > \sum_{j > i} \score_E(c_j)$.
Then $c \in \cR(E)$ if and only if $\score_E(c)\geq \score_E(c_i)$.
\end{Rule}

The candidates in Example~\ref{Exp1} together have 49 approvals. 
Therefore, a shortlist needs at least 25 approvals to be the First Majority winner set.
The smallest shortlist to achieve at least 25 approvals is $\{c_1, c_2,c_3\}$ with 29 approvals.

Next-$k$ \citep{Brams2012} is a rule that includes alternatives starting with the highest approval score, until a major drop in the approval scores is encountered, more precisely, if the total approval score of the next $k$ alternatives is less than the score of the previous alternative.

\begin{Rule}[Next-$k$]
Let $k$ be a positive integer. Then,
$c_i \in \cR(E)$ if for all $i' < i$ it holds that
$\score_E(c_{i'}) \leq \sum_{j=1}^{k} \score_E(c_{i'+j})$,
where $\score_E(c_{i' +j}) = 0$ if $i' +j> m$.
\end{Rule}

Consider Next-$2$. Then it is easy to check that, in Example~\ref{Exp1},
for all $i \leq 5$ the score of
$c_i$ is smaller or equal the sum of the scores of the next two candidates. 
For example $\score_E(c_1) = 10 \leq 19 = \score_E(c_2) + \score_E(c_3)$.
On the other hand $\score_E(c_7) = 3 \leq 0 = \score_E(c_8) + 0$.
Therefore, the winner set under Next-2 is $\{c_1, \dots, c_7\}$.

Observe that for both \Next{} and \First{} the winner set does not depend on the chosen enumeration of alternatives.
This will hold for all voting rules introduced in this paper.

\citet{Faliszewski17OA} discuss several other rules that satisfy
our basic axioms called Capped Satisfaction Approval Voting (CSA),
Net Approval Voting (NAV) and Net Capped Satisfaction Approval Voting (NCSA)
which were originally proposed by \citet{Brams2012} and \citet{BramsKilgour15} as 
well as generalizations of these rules. Among these \citet{Faliszewski17OA} conclude
that only the following generalization of NCSA is practical. 

\begin{Rule}[q-NCSA]
Let $q \in [0,1]$ be a real number and $S \subseteq C$ a set of candidates.
Then we define the \qNCSA-score of $S$ as:
\[\score_E^{\qNCSA}(S)=\sum_{v \in V}\left(\frac{\lvert S \cap A_v\rvert}{\lvert S\rvert^q} - \frac{\lvert S \cap \overline{A_v}\rvert}{\lvert S\rvert^q}\right).\]
The winner set then is the largest set with a maximum \qNCSA-score.\footnote{In the original definition of \citet{Faliszewski17OA}, all sets with maximum \qNCSA-score are co-winners. We specifically select the largest set with maximum \qNCSA-score as this choice guarantees that the winner set is non-tiebreaking (see discussion after Proposition~\ref{prop:qNCSA}).}
\end{Rule}

It is not immediately obvious that \qNCSA{} is a shortlisting rule in our sense. 
The following proposition shows that this is indeed the case and establishes some further
key properties of \qNCSA.

\begin{Prop}\label{prop:qNCSA}
The $q$-NCSA rule has the following properties for all $q \in [0,1]$.
\begin{enumerate}
\item It holds that $\score_E^{\qNCSA}(S)=\frac{1}{\lvert S\rvert^q}\sum_{c \in C}(2\score_E(c) - n)$.
\item $q$-NCSA is a shortlisting method.
\item $q$-NCSA can be computed in polynomial time.
\end{enumerate}
\end{Prop}
\begin{proof}
We prove the three statements separately:
\begin{enumerate}
\item Let $\mathbb{1}$ be the indicator function. Then we
can write the \qNCSA-score also as follows:
\begin{align*}
&\sum_{v \in V}\left(\frac{\lvert S \cap A_v\rvert}{\lvert S\rvert^q} - \frac{\lvert S \cap \overline{A_v}\rvert}{\lvert S\rvert^q}\right)
= \frac{1}{\lvert S\rvert^q} \sum_{v \in V}\left(\sum_{c\in S}\mathbb{1}_{c\in v} -
\sum_{c\in S}(1-\mathbb{1}_{c\in v})\right) =\\
&\frac{1}{\lvert S\rvert^q}\sum_{c \in S} \left(\sum_{v\in V}\mathbb{1}_{c\in v} -
\sum_{v\in V}(1-\mathbb{1}_{c\in v})\right) =
\frac{1}{\lvert S\rvert^q}\sum_{c \in S}(\score_E(c) - (n - \score_E(c))) =\\
& \frac{1}{\lvert S\rvert^q}\sum_{c \in S}(2\score_E(c) - n).
\end{align*}

\item It is clear that \qNCSA{} satisfies Anonymity and Neutrality. Consider Efficiency:
Assume there are two candidates $c_i$ and $c_j$ such that $\score_E(c_i) < 
\score_E(c_j)$, $c_i \in R(E)$ and $c_j \not\in R(E)$. Then there must be a $S \subseteq C$
with $c_i \in S$ and $c_j \not \in S$ which has maximal \qNCSA-score. However, by definition
the \qNCSA-score of $(S\setminus \{c_1\}) \cup \{c_j\}$ is higher than the \qNCSA-score of $S$.
A contradiction.

The non-tiebreaking property follows from the following claim:
\begin{claim}\label{claim1}
If $\score_E(c_i) = \score_E(c_{i+1})$
and \[\score_E^{\qNCSA}(\{c_1, \dots, c_{i-1}\}) \leq \score_E^{\qNCSA}(\{c_1, \dots, c_{i}\}),\]
then also 
\[\score_E^{\qNCSA}(\{c_1, \dots, c_{i}\}) \leq \score_E^{\qNCSA}(\{c_1, \dots, c_{i+1}\}).\]
\end{claim}
By Efficiency, the largest set with maximum \qNCSA-score is of the form $\{c_1, \dots,c_i\}$
for some $c_i$. As the set has maximum \qNCSA-score, it holds in particular that
$\score_E^{\qNCSA}(\{c_1, \dots, c_{i-1}\}) \leq \score_E^{\qNCSA}(\{c_1, \dots, c_{i}\})$.
If $\score_E(c_i) = \score_E(c_{i+1})$, i.e., if $\{c_1, \dots,c_i\}$ breaks a tie,
then it follows from the claim that $\score_E^{\qNCSA}(\{c_1, \dots, c_{i}\}) \leq
 \score_E^{\qNCSA}(\{c_1, \dots, c_{i+1}\})$.
However, this is a contradiction to the assumption that $\{c_1,\dots,c_i\}$ was the largest set
with maximal \qNCSA-score. The proof of Claim~\ref{claim1} contains a lengthy calculation and can be found in the appendix.

\item
As we have shown that \qNCSA is a shortlisting rules, we know that we only need to consider sets that are efficient and non-tiebreaking.
Further, we can clearly compute the \qNCSA-score of a set in polynomial time.
As there are at most linearly many potential winner sets (Observation~\ref{Obs1}), finding the one with maximum
\qNCSA-score can be done in polynomial time.
\end{enumerate}
\end{proof}

\noindent
Consider again Example~\ref{Exp1}. Then, the $0.5$-NCSA-score of the shortlist
$\{c_1, \dots, c_4\}$ is
$({(2\cdot10 -10) + (2\cdot10 -10) + (2\cdot9 -10) + (2\cdot8 -10)})/{\sqrt{4}} =  17.$
It can be checked that this is the unique maximal $0.5$-NCSA-score and hence
$\{c_1, \dots, c_4\}$ is the winner set under $0.5$-\textit{NCSA}.

\begin{Obs}\label{obs:qncsa}
An important feature (and downside) of \qNCSA{} is that candidates with an approval score of less than $\nicefrac n 2$ can only decrease $\score_E^{\qNCSA}(S)$. Consequently, \qNCSA{} returns the empty in elections where all candidates have few approvals.
\end{Obs}

\subsection{New Shortlisting Rules}

Let us now introduce some new shortlisting rules.
Similarly to \Next, the next two rules are based on the idea that one wants to make the cut
between winners and non-winners in a place where there is a large gap in
the approval scores. This can either be the overall largest gap or the first sufficiently large gap.

\begin{Rule}[Largest Gap]
Let $i$ be the smallest index such that 
$\score_E(c_{i}) - \score_E(c_{i+1}) = \max_{j < m} (\score_E(c_j) - \score_E(c_{j+1}))$.
Then $c \in \cR(E)$ if and only if $\score_E(c)\geq \score_E(c_i)$.
\end{Rule}

Note that in this definition a smallest index is guaranteed to exist due to our assumption that profiles are non-degenerate.
In Example~\ref{Exp1} the two largest gaps are between $c_5$ and $c_6$ and $c_7$ and $c_8$,
both of size $3$. As we pick the smaller index, the winner set
is $\{c_1, \dots, c_5\}$.

\begin{Rule}[First $k$-Gap]
Let $i$ be the smallest index such that 
$\score_E(c_i) - \score_E(c_{i+1}) \geq k$.
Then $c \in \cR(E)$ if and only if $\score_E(c)\geq \score_E(c_i)$.
If no such index exists, then $\cR(E) = C$, i.e., every alternative is a winner.
\end{Rule}

Let us consider \textit{First $2$-Gap} in Example~\ref{Exp1}.
The gaps between $c_1$ and $c_2$, $c_2$ and $c_3$ as well as between $c_3$ and $c_4$
are smaller than two, while the gap between $c_4$ and $c_5$ is $2$.
Therefore the winner set
is $\{c_1,c_2, c_3, c_4\}$.

The parameter $k$ has to capture what it means in a given shortlisting scenario
that there is a sufficiently large gap between alternatives, which
in particular depends on the number of voters $\lvert V\rvert$.
If no further information is available, one can choose $k$ by a
simple probabilistic argument.
Assume, for example, alternative $c$'s approval score is binomially distributed $\score_E(c)\sim B(n, q_c)$, where $n$ is the number of voters and $q_c$ can be seen as $c$'s quality.
We choose $k$
such that the probability of events of the following type 
are smaller than a selected threshold $\alpha$:
two alternatives $a$ and $b$ have the same objective quality ($q_a=q_b$) but have a difference in their approval scores of $k$ or more.
In such a case, the \Fgap{} rule might choose one alternative and not the other even though they are equally qualified, which is an undesirable outcome.
For example, if $n= 100$ and we want $\alpha = 0.5$, we have to choose $k \geq 5$
and if we want $\alpha = 0.1$ we need $k \geq 12$.
Note that this argument leads to rather large $k$-values; if further assumptions about the distribution of voters can be made, smaller $k$-values are feasible.

The voting rules above output winner sets of very different sizes (as we will see in the experimental evaluation, Section~\ref{sec:experiments}).
It is a common case, however, that there is a preferred size for the winner set,
but this size can be varied in order to avoid tiebreaking. This flexibility is especially 
crucial if the electorate is small and ties are more frequent.
Based on real-world shortlisting processes, we propose a rule that deals with this scenario by accepting
a preference order over set sizes as parameter and selecting a winner set with the most preferred size that does not require tiebreaking.

\begin{Rule}[Size Priority]
Let $\vartriangleright$ be a strict total order on $\{0, \dots, m\}$, the \emph{priority order}.
Then $\cR(E)=\{c_i\in C \mid 1\leq i \leq k\}$ if and only if
\begin{itemize}
\item either $\score_E(c_{k}) \neq \score_E(c_{k +1})$ or $k = 0$ or $k = m$,%
\item and $\score_E(c_{\ell}) = \score_E(c_{\ell+1})$ for all $\ell \vartriangleright k$.
\end{itemize}
\end{Rule}

Consider for example a strict total order of the form
$1 \vartriangleright 6 \vartriangleright 0 \vartriangleright \dots$.
Then the set of \SPAV{} winners under $\vartriangleright$ in Example~\ref{Exp1} is the empty set,
because $\{c_1\}$ and $\{c_1, \dots, c_6\}$ break ties, as $\score_E(c_1) = \score_E(c_2)$ 
and $\score_E(c_6) = \score_E(c_7)$.

\SPAV{} is a non-tiebreaking analogue of \kAV{}, which selects the $k$ alternatives with the highest approval score.
A specific instance of \SPAV{} was used by the Hugo Award prior to 2017
with the priority order $5 \vartriangleright 6 \vartriangleright 7 \dots$~\citep{Hugo}.
Generally, the choice of a priority order depends on the situation at hand. For award-shortlisting, typically a small number of alternatives is selected (the Booker Prize, e.g., has a shortlist of size 6).
In a much more principled fashion, \citet{amegashie1999design} argues that the optimal size of the winner set for shortlisting should be proportional to $\sqrt{m}$, i.e., the square root of the number of alternatives.

In practice, the most common priority order is $k \vartriangleright k+1 \vartriangleright 
\dots \vartriangleright m$ for some $k < m$, i.e., the smallest non-tiebreaking shortlist that contains at least $k$ 
alternatives is selected.
Another important special case are instances of \SPAV{} that rank $0$ and $m$ the lowest,
i.e., that are decisive whenever possible.
Therefore, we give \SPAV{}
rules with based on such priority orders a special name.

\begin{Def}
Let $\vartriangleright$ be a strict total order on $0, \dots, m$
and let $k$ be a positive integer with $k \leq m$
such that $k \vartriangleright k+1 \vartriangleright 
\dots \vartriangleright m$ and $m \vartriangleright \ell$ for all $\ell < k$.
Then, the \SPAV{} rule defined by the priority order $\vartriangleright$ is an \ISPAV{} rule.
We will write \ispav{k} as a short form for the \ISPAV{} rule with $k \vartriangleright k+1 \vartriangleright \dots$ as priority order.

Let $\vartriangleright$ be a strict total order on $0, \dots, m$
such that $k \vartriangleright m$ and $k \vartriangleright 0$ holds for all $0 < k < m$.
Then, the \SPAV{} rule defined by the priority order $\vartriangleright$ is a \DSPAV{} rule.
\end{Def}
 
Other special cases of \SPAV{} could be defined in a similar way, for example
\textit{Decreasing Size Priority}. However, \ISPAV{} and \DSPAV{} are the most natural and common
types of \SPAV{} and additionally satisfies better axiomatic properties than, e.g.,
\textit{Decreasing Size Priority}.

Finally, we propose a rule that combines the ideas behind \Fgap{} and \SPAV.
In practice, we often want to have a large gap between winners and non-winners, but not 
at any price in terms of the size of the shortlist.

\begin{Rule}[Top-$s$-First-$k$-gap]
Let $W'$ be the winner set for \Fgap{} and $W''$
the winner set for the \ISPAV{} instance defined by the order
$s \vartriangleright s+1 \vartriangleright \dots$ (\ispav{s}). If $\lvert W'\rvert \leq s$,
return $W'$. Otherwise return $W''$.
\end{Rule}

Consider, for example, \TopFgap{3}{2}. Then, in Example~\ref{Exp1} we know that
$W' = \{c_1, \dots, c_4\}$ is the \textit{First-2-gap} winner set. On the other hand,
the shortlist $W'' = \{c_1,c_2,c_3\}$ is non-tiebreaking and therefore the \SPAV{}
winner set for $3 \vartriangleright 4 \vartriangleright \dots$. As $\lvert W'\rvert > s$, the winner set under $\TopFgap{3}{2}$ is $W''$.

Let us now consider the relationships
between the proposed rules. 

\begin{Prop}
We observe the following relations between the considered voting rules:
\begin{itemize}
\item \Fgap{} and \Next{} are equivalent to \Approval{} for $k =1$.
\item \ispav{1} is equivalent to \Approval.
\item \TopFgapSK{} is equivalent to \Fgap{} for $s = m$ and it is equivalent to \ISPAV{} for 
$k = m$. 
\end{itemize}
\end{Prop}

\begin{proof}
First observe that \textit{First-$1$-Gap} and \textit{Next-$1$}
select all candidates which have maximal score.
Now let $c_i$ be the first candidate which has less than the maximal score. Then $\score_E(c_{i-1}) - \score_E(c_i) \geq 1$ and  thus \textit{First-$1$-Gap} selects $\{c_1,\dots,c_{i-1}\}$ (as does \Approval).
Further, $\score_E(c_{i-1} > \score_E(c_i) = \sum_{j=1}^1 \score_E(c_{i-1+j})$ and thus \textit{Next-$1$} selects $\{c_1,\dots,c_{i-1}\}$.
The argument for \ispav{1} is similar.

Finally, consider \TopFgapSK. If $s = m$, then $W'' = C$. Consequently, $\lvert W'\rvert\leq \lvert W''\rvert$ and $W'$ is thus the winner set.
On the other hand, if $k = m$ then $W' = C$. Hence, $\lvert W''\rvert\leq \lvert W'\rvert$ and $W''$ is thus the winner set.
\end{proof}

We finally observe that \qNCSA{} for $q=1$ is a mix of \Approval{} and \Threshold{}
and for $q= 0$ is closely related to \Threshold{} for $f(n) = \frac{1}{2}n$.
First consider $q= 1$.
If any candidate is approved by more than $50\%$ of the voters then \textit{1-NCSA}
is equivalent to \Approval, as the \textit{1-NCSA}-score equals the average net-approval of the
candidates in the set. This score is maximized by any set only containing candidates
with maximal approval. On the other hand, if no candidate has more than $50\%$ approvals
then no set has positive \qNCSA-score. Therefore, the empty set is the smallest set with
maximal \qNCSA-score.

Now consider $q=0$. We observe that then \qNCSA-score of a set $S \subseteq C$
is the sum of the \emph{net-approval} of the candidates, where the net approval
of a candidate $c$ is $\score_E(c) - (n -\score_E(c))$. Hence the \textit{0-NCSA}-score
is maximized by every set that contains all candidates with
positive net approval and an arbitrary number of candidates with $0$ net approval. 
A candidate has non-negative net approval if and only if $2\score_E(c) -n \geq 0$
which is equivalent to $\score_E(c) \geq \frac{n}{2}$.\footnote{This is not equivalent
to \Threshold{} for $f(n) = \frac{1}{2}n$, as candidates with exactly $\frac{1}{2}n$
approval score are not included in the winner set of \Threshold.}

To conclude the section, let us remark that all of the above rules can be computed in polynomial time. This follows immediately from their respective definitions. For \qNCSA{}, we made the argument explicit in Proposition~\ref{prop:qNCSA}.

\section{Axiomatic Analysis}\label{sec:axiomatic}

In this section, we axiomatically analyze shortlisting rules with the goal to discern their defining properties.
First, we consider axioms that are motivated by 
the specific requirements of shortlisting,  
then we study well-known axioms that describe
more generally desirable properties of voting rules.
For an overview, see Table~\ref{tab:Results}.

\begin{table*}
\centering
\small
    \begin{tabular}{l*{9}c}
         & \rot{Unanimity} & \rot{Anti-Unanimity} & 
         \rot{$\ell$-Stability} &
         \rot{Determined} &
        \rot{Independence} &
         \rot{Ind. of Losing Alt.} &
        \rot{Res.\ to Clones} &
        \rot{Set Monot.} &
        \rot{Superset Monot.} %
\\
        \midrule
         Approval Voting       & \OK     & \OK  & \NO          &\OK    & \NO & \OK & \OK  & \OK & \OK\\ %
         $f$-Threshold         & \OK     & \OK  & \NO          &\NO    & \OK & \OK & \OK  & \OK & \NO\\ %
       Max-Score-$f$-Threshold & \OK     & \OK  & \NO          &\OK    & \NO & \OK & \OK  & \OK & \NO\\ %
         First Majority        & \OK     & \OK  & \NO          &\OK    & \NO & \NO & \NO  & \NO & \NO\\ %
         $q$-NCSA              & \OK     & \OK  & \NO          &\NO    & \NO & \OK & \NO  & \OK & \NO  \\ %
         Next-$k$              & \OK     & \OK  & \NO          &\OK    & \NO & \NO & \NO  & \OK & \NO\\ %
         Largest Gap           & \OK     & \OK  & \NO          &\OK    & \NO & \NO & \OK  & \OK & \NO\\ %
         First $k$-Gap         & \OK     & \NO  &$\ell\leq k$  &\OK    & \NO & \OK & \OK  & \OK & \OK\\ %
         Decis. Size Priority  & \OK     & \OK  & \NO          &\OK    & \NO & \NO & \NO  & \OK & \NO\\ %
        Incr. Size Priority    & \OK     & \NO  & \NO          &\OK    & \NO & \OK & \NO  & \OK & \OK\\ %
        Top-$s$-First-$k$-Gap  & \OK     & \NO  & \NO          &\OK    & \NO & \OK & \NO  & \OK & \OK\\
        \midrule
    \end{tabular}
    \caption{Results of the axiomatic analysis.}
    \label{tab:Results}
\end{table*}

\subsection{$\ell$-Stability, Unanimity, and Anti-Unanimity}

When shortlisting is used for the initial screening of candidates,
for example for an award or a job interview, then we cannot assume that the voters 
have perfect judgment. Otherwise, there would be no need for a second round of deliberation,
as we could just choose the highest-scoring alternative as a winner.
Therefore, small differences in approval may not correctly reflect
which alternative is more deserving of a spot on the shortlist.
Thus, out of fairness, we want our voting rule to treat alternatives 
differently only if there is a significant difference in approval between them. 

\begin{Axm}[$\ell$-Stability]
If the approval scores of two alternatives differ by less than $\ell$,
either both or neither
should be a winner, i.e., 
for every election $E=(C,V)$ and candidates $c_i$ and $c_j$
if $\lvert \score_E(c_i) - \score_E(c_j)\rvert < \ell$
then either $c_i,c_j \in \cR(E)$ or $c_i,c_j \not \in \cR(E)$.
\end{Axm}

Here, the parameter $\ell$ has to capture what constitutes a significant difference in 
a given election. This will depend, for example, on the number and trustworthiness of the voters.
Also, observe that $1$-Stability equals Non-tiebreaking.

Now, while a small difference in approvals might not correctly reflect the relative quality 
of the candidates, we generally assume in shortlisting that the approval scores
approximate the underlying quality of alternatives\footnote{The relation between
approval voting and maximum likelihood estimation is analyzed in
detail by Procaccia and Shah~\citep{procaccia2015approval},
in particular, under which conditions approval voting selects the most likely
``best'' alternatives.}. Therefore, at a minimum, we want to include alternatives that are approved
by everyone and exclude alternatives that are approved by no one.

\begin{Axm}[Unanimity] If an alternative is approved by everyone, it must be a winner,
 i.e., 
for every election $E=(C,V)$, if
$\score_E(c) = n$ then $c \in \cR(E)$.
\end{Axm}

\begin{Axm}[Anti-Unanimity] If an alternative is approved by no one, it cannot win,
 i.e.,
for every election $E=(C,V)$ if  
$\score_E(c) = 0$ then $c \not\in \cR(E)$.
\end{Axm}

Unfortunately, it turns out that these three axioms are incompatible
unless there are many more voters than alternatives.
Indeed Unanimity, Anti-Unanimity and $\ell$-Stability can be jointly satisfied
if and only if $\lvert V\rvert \geq l\cdot\lvert C\rvert +1$. 

\begin{Thm}\label{Prop:Size-Stable}
For every $\ell$ there is a shortlisting rule that satisfies Unanimity, Anti-Unanimity and $\ell$-Stability
for every election $E$ such that $n_E > (\ell - 1)\cdot(m_E -1)$.
This is a tight bound in the following sense:
For every $\ell>1$, there is an election $E$ such that $n_E = (\ell - 1)\cdot(m_E -1)$
and no shortlisting rule can satisfy Unanimity, Anti-Unanimity and $\ell$-Stability
on $E$. 
\end{Thm}

\begin{proof}
To show that Unanimity, Anti-Unanimity and $\ell$-Stability
are jointly satisfiable if $n_E > (\ell - 1)\cdot(m_E -1)$,
we will show that a slightly modified version of \Fgap{} satisfies 
all three axioms for elections $E$ with $n_E > (\ell -1) \cdot(m_E - 1)$.
We define \MFgap{} as follows:
Let $c_1, \dots, c_m$ be an enumeration of $C$
such that $\score_E(c_{i-1}) \geq \score_E(c_{i})$.
Let $i$ be the smallest index such that 
$\score_E(c_i) - \score_E(c_{i+1}) \geq \ell$.
Then $c \in \cR(E)$ if and only if $\score_E(c)\geq \score_E(c_i)$.
If no such index exists, then 
$\cR(E) = \emptyset$ if there is an alternative $c$ with $\score_E(c) = 0$, and 
$\cR(E) = C$ otherwise.
Clearly, this rule still satisfies $\ell$-Stability.

Now, let $E$ be an election such that there is an alternative $c$ 
with $\score_E(c) = n$. 
Assume first that there is no alternative $c'$ with $\score_E(c') = 0$.
In that case, \MFgap{} vacuously satisfies Anti-Unanimity and, by definition, %
also Unanimity.
Now assume that there is an alternative $c$ with $\score_E(c) = 0$.
We claim that there is an index $i$ such that 
$\score_E(c_i) - \score_E(c_{i+1}) \geq \ell$
and hence only alternatives $c$ such that $\score_E(c) \geq \score_E(c_i)> \ell -1$
are winners.
Otherwise, we have 
$\score_E(c_{i+1}) \geq \score_E(c_{i}) - (\ell -1) $
for all $i < m$ and hence  
$\score_E(c_m) \geq \score_E(c_{1}) -(\ell -1)\cdot (m-1)$.
However, as $\score_E(c_1) = n > (\ell -1)\cdot(m-1)$ this contradicts 
the assumption that there is an alternative $c$ with $\score_E(c) = 0$,
i.e., $\score_E(c_m)=0$.

Finally, let $E$ be an election such that there is no alternative $c$ with $\score_E(c) = n$.
Then, \MFgap{} vacuously satisfies Unanimity. 
Now, if there is an alternative $c'$ with $\score_E(c') = 0$ then 
we have to distinguish two cases. If there is no $\ell$-gap, then
$\cR(E) = \emptyset$ by definition and hence \MFgap{} satisfies Anti-Unanimity.
On the other hand, if there is a $\ell$-gap, then only alternatives above the $\ell$-gap 
are selected, which must have a score of $\ell$ or larger. Hence,
Anti-Unanimity is also satisfied.

Now we show the tightness of the theorem.
Let $E$ be an election with 2 alternatives and $\ell -1$ voters such that
$\score(E) = (\ell -1,0)$. We observe $n_E = \ell -1 = (\ell-1)\cdot(2-1)$.
We claim that no $\cR$ satisfy Unanimity, Anti-Unanimity and 
$\ell$-Stability on $E$.
Hence, $c_1 \in \cR(E)$ must hold by Unanimity.
Then $\score_E(c_1) - \score_E(c_2) < \ell$ implies
$c_2 \in \cR(E)$ by $\ell$-Stability, contradicting Anti-Unanimity.
\end{proof}

Theorem~\ref{Prop:Size-Stable} tells us that $\ell$-Stability requires some sacrifices as 
it is incompatible with the combination of Unanimity and Anti-Unanimity.
However, \Fgap{} can be seen as an optimal compromise 
as, with a small modification, it satisfies
Anti-Unanimity whenever Theorem~\ref{Prop:Size-Stable} allows it.

Let us now analyze the considered shortlisting rules with regard to the three axioms  
Unanimity, Anti-Unanimity and $\ell$-Stability.

\begin{itemize}
\item  It is straightforward to see that \Approval, \Threshold, \MSThreshold{} and \Lgap{} 
satisfy Unanimity and Anti-Unanimity for all non-degenerate profiles.
Hence, they cannot satisfy $\ell$-Stability for $\ell > 1$.

\item By definition, \Fgap{} satisfies Unanimity and $\ell$-Stability for
$k\geq \ell$ for all elections. Therefore, it cannot satisfy Anti-Unanimity.

\item \First{} satisfies Unanimity as $c_1 \in \cR(E)$ by definition and $\score_E(c_i) =
\max_{c\in C}(\score_E(c)) = \score_E(c_1)$ implies $c_i \in \cR(E)$.
Furthermore, we claim that it satisfies Anti-Unanimity. Let $c_1, \dots, c_m$ be the enumeration
of the candidates used by \First{} and let $c_i$ be the first candidate with $\score_E(c_i) =0$.
Then, $\sum_{j >(i-1)} \score_E(c_{j}) > 0$ while $\sum_{j >(i-1)} \score_E(c_{j}) = 0$.
This implies that $c_i \not \in \cR(E)$.  
It follows that \First{} does not satisfy $\ell$-Stability for $\ell>1$.

\item \Next{} satisfies Unanimity by definition. Furthermore, we claim that
it satisfies Anti-Unanimity. Let $c_1, \dots, c_m$ be the enumeration
of the candidates used by \Next{} and let $c_i$ be the first candidate with $\score_E(c_i) =0$.
Then, $\score_E(c_{i-1}) > 0$ while $\sum_{j = 1}^{k} \score_E(c_{(i-1)+j}) = 0$.
This implies that $c_i \not \in \cR(E)$.  
As before, it follows that \Next{} does not satisfy $\ell$-Stability for $\ell>1$.

\item \qNCSA{} satisfies Unanimity and Anti-Unanimity. 
First, we show that \qNCSA{} satisfies Unanimity: Let $c_i$ be a candidate
with $\score_E(c_i) = n$. Then, $\score_E^{\qNCSA}(\{c_i\}) > 0 = \score_E^{\qNCSA}(\emptyset)$ 
and therefore $\cR(E) \neq \emptyset$. As $\score_E(c_i) = \max_{c\in C}(\score_E(c))$
this implies by Efficiency and Non-tiebreaking that $c_i \in \cR(E)$.
Now, we show that \qNCSA{} satisfies Anti-Unanimity: Let $c_i$ be a candidate
with $\score_E(c_i) = 0$. Then $2\score_E(c_i) -n < 0$, which means for every 
set~$S$ such that $c_i \in S$ that the \qNCSA-score of $S$ is strictly
smaller than the \qNCSA-score of $S \setminus\{c_i\}$. This means that 
the \qNCSA-score of $S$ is not maximal. As $S$ was chosen arbitrarily, we can conclude
$c_i \not \in \cR(E)$.
As above, it follows that \Next{} does not satisfy $\ell$-Stability for $\ell>1$.

\item Now, we claim that \SPAV{} always satisfies either Unanimity or Anti-Unanimity:
First, we show that it satisfies Unanimity if $m \vartriangleright 0$:
As selecting all $m$ candidates can never be tie-breaking, \SPAV{}
will never select the empty set in this case. This implies
$c_1 \in \cR(E)$ for all $E$. As \SPAV{} is non-tiebreaking it follows
that all $c_i$ with $\score_E(c_i)= n$ must be in the winning shortlist.
On the other hand, \SPAV{} satisfies Anti-Unanimity if $0 \vartriangleright m$ holds
by a symmetric argument. 

Moreover, we claim that it satisfies both axioms
(for non-degenerate profiles) if and only if it is a \DSPAV{} rule.
First let $\cR$ be a \DSPAV{} rule.
Then, for every non-degenerate profile there must be a $i$
such that $c_1, \dots, c_i$ can be selected as winners without tiebreaking.
As $\cR$ is a \DSPAV{} rule, we know $i \vartriangleright m$ and $i \vartriangleright 0$.
It follows that $\cR(E)$ is neither $C$ nor $\emptyset$. By a similar argument as before,
this implies that $\cR$ satisfies Unanimity and Anti-Unanimity.

Now assume that $\cR$ is a \SPAV{} instance that is not determined. 
Then there exists a $0 <k <m$ such that either $m \vartriangleright k$ or
$0 \vartriangleright k$. Consider an election $E$
such that $\score_E(c_i) = n$ for all $i \leq k$ and $\score_E(c_j) =0$
if $j > k$. Then, in the first case all candidates are winners and hence Anti-Unanimity
is violated and in the second case no candidate is a winner and hence Unanimity is 
violated.

It follows from the above that \ISPAV{} satisfies 
Unanimity but not Anti-Unanimity.
Finally, \SPAV, by definition, satisfies $\ell$-Stability for $\ell > 1$ if and only if
$0$ or $m$ is the most preferred size.
\item Finally, \TopFgapSK{} satisfies Unanimity because \Fgap{} and \ISPAV{} do so. 
That means for all election both winner sets $W'$ and $W''$ considered by \TopFgapSK{}
satisfy unanimity. Hence, whatever set is chosen, unanimity is satisfied.
On the other hand, it satisfies neither $\ell$-Stability for $\ell > 1$ nor Anti-Unanimity.
Consider first the election given by $\score(E) = (3,2,0,0)$ and \TopFgap{1}{2}.
Then $\{c_1\}$ is the winner set and hence $2$-Stability is violated.
Now consider the same election under \TopFgap{3}{3}. Then both $W'$ and $W''$ equal
$\{c_1, \dots,c_4\}$ and hence Anti-Unanimity is violated.

\end{itemize}

We observe that \Fgap{} is the
only voting rule considered in this paper that satisfies $\ell$-Stability for $\ell>1$:
However, it is worth noting that \Lgap{} satisfies $\ell$-Stability whenever 
there is an $\ell$-gap.

\subsection{Minimal Voting Rules}

The goal of shortlisting is to reduce a set of alternatives 
to a more manageable set of alternatives. 
It is therefore desirable that shortlisting rules produce short shortlists,
without compromising on quality. To formalize this desideratum
we define the concept of a minimal voting rule that satisfies a set of axioms.

\begin{Def}
Let $\mathcal{A}$ be a set of axioms and let $S(\mathcal{A})$
be the set of all voting rules satisfying
all axioms in $\mathcal{A}$.
Then, we say a voting rule is a minimal voting rule $\cR$
for $\mathcal{A}$
if for all elections $E$ it holds that
$\cR(E) = \bigcap_{\cR^* \in S(\mathcal{A})}  \cR^*(E)$.
\end{Def}

We observe that in general a minimal voting rule $\cR$
for a set of axioms $\mathcal{A}$ does not satisfy
all axioms in $\mathcal{A}$.
Consider, e.g., the following axiom:

\begin{Axm}[Determined]
Every election must have at least one winner, i.e.,
for all elections $E$ we have $\cR(E) \neq \emptyset$.
\end{Axm}

First, observe that
besides \Threshold{}, \qNCSA{} and \SPAV{} all voting rules
considered in this paper are determined by definition.
For \Threshold{} it is clear that $\cR(E)$ can be empty if 
no candidate achieves enough approvals to clear the threshold.
Observe that this is not the case for \MSThreshold,
as we assume $f(n) < n$ and hence candidates with maximal
score are always winners.
\SPAV{} returns the empty set if $0$ is the most preferred
set size that does not require tiebreaking.
This cannot happen if \SPAV{} is either a \DSPAV{} rule
or $m \vartriangleright 0$; \SPAV{} is determined in these cases. In particular
this means that \ISPAV{} is determined. 
For \qNCSA{} we observe that if no candidate has at least $50\%$
approvals then $2\score_E(c) - n$ is negative for all candidates 
and hence the \qNCSA-score is only maximized by the empty set. Hence
\qNCSA{} is not determined.

Now, let us consider arbitrary voting rules with a variable number of winners, i.e.,
not only shortlisting rules. Then for every $c \in C$ the rule $\cR_c$
that always outputs the set~$\{c\}$ is a determined voting rule.
It follows that the minimal determined voting rule always outputs the empty set
and is hence not determined.
In contrast, for shortlisting rules the following holds.

\begin{Prop}
Let $\mathcal{A}$ be a set of axioms that contains the four basic shortlisting axioms
(Axioms 1--4). Then the minimal voting rule for $\mathcal{A}$ is again a shortlisting rule, i.e., it satisfies Axioms 1--4.
\end{Prop}

\begin{proof}
Let $\mathcal{A}$ be a set of axioms and let $\cR$
be the minimal voting rule for $\mathcal{A}$.
It is straightforward to see that $\cR$ satisfies Neutrality and Anonymity.
We show that $\cR$ also satisfies Efficiency and is non-tiebreaking.
Let $E$ be an election. 
As every rule in $S(\mathcal{A})$ is a shortlisting rule,
there is a $k_{\cR^*} \in \{0,\dots, m\}$ for every rule $\cR^* \in S(\mathcal{A})$
such that $\cR^*(E) = \{c_1, \dots , c_{k_{\cR^*}}\}$.
Now let $k_m$ be the smallest $k$ 
such that there is a rule $\cR^* \in S(\mathcal{A})$
with $\cR^*(E) = \{c_1, \dots, c_{k}\}$. 
Then, by definition $\cR(E) = \cR^*(E)$.
As $\cR^*(E)$ does not violate Efficiency and non-tiebreaking for $E$, neither does $\cR$. 
As this argument holds for arbitrary elections, $\cR$ satisfies Efficiency and is non-tiebreaking.
\end{proof}

As the voting rule that always outputs the empty set is a shortlisting rule, it is also the 
minimal shortlisting rule (without additional axioms). Therefore, we need to assume
additional axioms. We consider determined and $\ell$-stable shortlisting rules.

\begin{Thm}
\Approval{} is the minimal voting rule that is
efficient, non-tiebreaking and determined.
Furthermore, for every positive integer $k$, \Fgap{} is the minimal voting rule that is
efficient, $k$-stable and determined.
\end{Thm}

\begin{proof}
Let $\mathcal{A}$ be the set
$\{\mbox{Efficiency}, k\mbox{-Stability}, \mbox{Determined}\}$
and $\cR$ be \Fgap.
We know that \Fgap{} is efficient, $k$-stable and determined,
therefore we know 
$\bigcap_{\cR^* \in S(\mathcal{A})}  \cR^*(E) \subseteq \cR(E)$.

Now, every determined voting rule must have a non-empty set of winners.
If the voting rule is efficient, the set of winners must contain at least one top ranked 
alternative. 
Now, consider an enumeration of the alternatives
$c_1, \dots ,c_m$ such that $\score_E(c_j) \geq \score_E(c_{j+1})$
holds for all $j$. If a voting rule is $k$-stable, a winner set 
containing one top ranked alternative must contain all alternatives $c_i$ for which 
$\score_E(c_j) < \score_E(c_{j+1}) +k$ holds for all $j < i$. 
By the definition of \Fgap{} this implies
$\cR(E) \subseteq \bigcap_{\cR^* \in S(\mathcal{A})}  \cR^*(E)$.

The minimality of \Approval{} is a special case of the minimality of \Fgap,
as $1$-Stability equals Non-tiebreaking and \textit{First-1-Gap} is equivalent 
to \Approval.
\end{proof}

This result is another strong indication that \Fgap{} is promising from an axiomatic standpoint.
It produces shortlists that are as short as possible without 
violating $k$-Stability, an axiom that is desirable in many shortlisting 
scenarios.

Next we will consider axioms that are not specific to shortlisting, 
but often appear in the voting and judgment aggregation literature 
to characterize ``well behaved'' aggregation techniques.

\subsection{Independence}

$\ell$-Stability
formalizes the idea that the length of a shortlist should take 
the magnitude of difference between approval scores into account.
This contradicts an idea that is often considered in judgment 
aggregation, namely that all alternatives should be treated independently~\citep{Handbook-JA}.

\begin{Axm}[Independence]
If an alternative is approved by exactly the same voters in two elections
then it must be a winner either in both or in neither.
That is,
for an alternative $c$, and two elections $E = (C,V)$ and $E^* = (C,V^*)$ with $\lvert V\rvert = \lvert V^*\rvert$ and
 $c \in v_i$ if and only if $c \in v_i^*$ for all $i\leq n$, it holds that $c \in \cR(E)$ if and only if $c \in \cR(E^*)$.
\end{Axm}

\Threshold{} rules are the only rules in our paper satisfying Independence.
Indeed, Independence characterizes \Threshold{} rules.

\iffalse
\begin{Axm}[Weak Local Monotonicity]
Let $E = (C,V = (v_1, \dots, v_n))$ be an election.
Assume $c_i \in \cR(E)$.
If $E^* = (C,V^* = (v^*_1, \dots, v^*_n))$  is another election such that for some $j \leq n$
we have $v^*_j  = v_j \cup \{c_i\}$ and $v^*_l = v_l$ for all $l \neq j$,
then $c_i \in \cR(E^*)$.
\end{Axm}

This is a very weak axiom that is satisfied by all voting rules that we consider.
\fi

\begin{Thm}\label{Class:Threshold-thm}
Given a fixed set of alternatives~$C$, every shortlisting rule that satisfies Independence
is an \Threshold{} rule for some function $f$.
\end{Thm}

\begin{proof}
Let $\cR$ be a voting rule that satisfies Anonymity and Independence.
Then we claim that for two elections $E = (C,V)$ and $E^* = (C,V^*)$ with $\lvert V\rvert = \lvert V^*\rvert$
and an alternative $c_i \in C$ we have that $\score_E(c_i) = \score_{E^*}(c_i)$
implies that either $c_i \in \cR(E),\cR(E^*)$ or
$c_i \not \in \cR(E),\cR(E^*)$.
If $\score_E(c_i) = \score_{E^*}(c_i)$, then there is a 
permutation $\pi: \{1,\dots,n\} \to \{1,\dots,n\}$
such that $c_i \in v_i$ if and only if $c_i \in v^*_{\pi(i)}$.
Now, let $E' = (C, \pi(V))$. Then, by Anonymity, 
$c_i \in \cR(E)$ if and only if $c_i \in \cR(E')$. 
Now, as $c_i$ is approved by the same voters in
$E'$ and $E^*$, Independence implies $c_i \in \cR(E')$ if and only if
$c_i \in \cR(E^*)$.
  
Now, let $E = (C,V)$ and $E^* = (C,V^*)$ be two elections with $\lvert V\rvert = \lvert V^*\rvert$.
Furthermore, assume $c_i \in \cR(E)$ and $\score_E(c_i) < \score_{E^*}(c_i)$.
We claim that this implies $c_i \in \cR(E^*)$.
By Independence, we can assume w.l.o.g.\ that there is an alternative $c_j$ such 
that $\score_E(c_j) = \score_{E^*}(c_i)$.
Then, by Efficiency, $c_j \in \cR(E^*)$. Now, let $E'$
be the same election as $E$ but with $c_i$ and $c_j$ switched.
Then by Neutrality we have $c_i \in \cR(E')$.
As by definition $\score_{E'}(c_i) = \score_{E^*}(c_i)$
this implies $c_i \in \cR(E^*)$ by Anonymity and Independence.

The two arguments above mean that for every alternative $c_i$ and $n \in \mathbb{N}$ there is a $k$
such that for all elections $E = (C,V)$ with $\lvert V\rvert = n$
we know $c_i \in \cR(E)$ if and only if $\score_E(c_i) \geq k$.
If $\cR$ also satisfies Neutrality, then $k$ must be the same for every $c_i \in C$
and hence $\cR$ must be a Threshold rule.
\end{proof}

In light of Theorem~\ref{Class:Threshold-thm}, Independence seems to be a very strong requirement,
therefore we also consider the axiom Independence of Losing
Alternatives which can be seen as a weakening of Independence.
It states that removing a non-winning alternative cannot
change the outcome of an election.

\begin{Axm}(Independence of Losing Alternatives)
Let $E = (C,V)$ with $V=(v_1,\dots,v_n)$ and $E^* = (C^*, V^*)$
where $C^* = C \setminus \{c^*\}$ and
$V^*=(v_1^*,\dots,v_n^*)$ be two elections such that $c^* \not \in \cR(E)$
and $v_i^* = v_i \setminus \{c^*\}$ for all $i\leq n$. Then $\cR(E) = \cR(E^*)$.
\end{Axm}

Clearly, \Threshold{} satisfies this axiom as it also satisfies Independence.
As removing a losing alternative does not change the maximal score, the
same holds for \MSThreshold.
Furthermore, as the removal of a losing alternative can only widen the gap between
the winners and the non-winners, \Fgap{} satisfies Independence of Losing Alternatives, and
so does \Approval{}, which is a special case of \Fgap{}.
Finally, for \qNCSA{} the removal of a losing alternative just removes
some non-maximal sets from consideration. Clearly, this does not 
change which sets have maximal \qNCSA-score.

None of the other rules satisfy Independence of Losing Alternatives.
\begin{itemize}
\item \First{}: Assume
$E$ is an election such that $\score(E) = (3,2,1,0)$.
Then the winner set under \First{} is $\{c_1,c_2\}$ but removing 
$c_3$ changes the winner set to $\{c_1\}$.
\item \Lgap: Consider the same election as for \First. Then, the winner set under \Lgap{}
is $\{c_1\}$ but removing $c_3$ changes this to $\{c_1,c_2\}$.
\item \Next: Consider an election $E$ with $\score(E) = (4,3,2,0)$.
Then, for every $k > 1$, we have $\cR(E) = \{c_1,c_2\}$ under \Next{},
but after deleting $c_3$ we have $\cR(E) = \{c_1\}$.
\end{itemize}

For \SPAV{} we encounter a difficulty: Independence of Losing Alternatives cannot be applied to \SPAV{}
because each instance of \SPAV{} is defined by a linear order on $0, \dots, m$ and
decreasing the number of alternatives necessitates a different order.
We can deal with this problem by defining classes of \SPAV{} instances:

\begin{Def}
Let $\vartriangleright$ be a linear order on~$\mathbb{N}$.
Then the class of \SPAV{} instances defined by $\vartriangleright$ contains 
for every number of alternatives~$m$ the \SPAV{} instance
given by the restriction of $\vartriangleright$ to $\{0,1,\dots,m\}$.

We say that the class of \SPAV{} instances defined by $\vartriangleright$
is a class of \ISPAV{} instances if every \SPAV{} instance in the class is an \ISPAV{}
instance.
\end{Def}

This definition allows us to ask whether classes of \SPAV{} instances
(defined by $\vartriangleright$) satisfies Independence of Losing Alternatives.
Consider, e.g., the class of \SPAV{} instances
defined by any order of the form $2 \vartriangleright 1 \vartriangleright \dots$
and an election $E$ with $\score(E) = (2,1,1)$.
Then $\cR(E) = \{c_1\}$ but the removal of $c_3$ leads to $\cR(E) = \{c_1,c_2\}$.
Thus, \SPAV{} fails Independence of Losing Alternatives in general.
However, we claim that every class of \ISPAV{} instances satisfies
Independence of Losing Alternatives.
We distinguish two cases: First assume all $m$-alternatives are selected.
Then Independence of Losing Alternatives is vacuously satisfied as there are no
losing alternatives. On the other hand, assume that there is a $k < m$
such that $\{c_1, \dots,c_k\}$ is winning. As $\cR$ is an \ISPAV{} instance
there is a $k' \leq k$ such that $\vartriangleright$ restricted to $\{1, \dots, m\}$
starts with $k' \vartriangleright k'+1 \vartriangleright \dots \vartriangleright k$.
As $k < m$ the same holds for $\vartriangleright$ restricted to $\{1, \dots, m-1\}$.
Hence if we remove an alternative $c_{k^*}$ with $k^* > k$ the winner set does not change.

Finally, we claim that \TopFgapSK{} also satisfies Independence of Losing Alternatives. 
Assume first that the set of winners $W'$ under \Fgap{} is smaller than $s$. After removing
a losing alternative, the set of winners under \Fgap{} remains the same and is hence still smaller 
then $s$. It follows that the winner set under \TopFgapSK{} does not change.
Now assume that $W'$ is larger than $s$. Then, the winner set of \TopFgapSK{}
has size at least $s$. Now, removing an alternative $c_j$ for $j > s$ cannot create a larger gap
between the first $s$ alternatives. It follows that the winner set under \Fgap{} after removing
$c_j$ is still larger then $s$. This means by definition that the winner set of \TopFgapSK{}
before and after removing $c_j$ was the winner set of \ISPAV. As \ISPAV{}
satisfies Independence of Losing Alternatives, we can conclude that \TopFgapSK{} does so as well.

\subsection{Further Axioms}
Finally, we consider three classic axioms of social choice theory, namely Resistance to Clones
\citep{tideman1987independence} and two monotonicity axioms \citep{Handbook-Voting}
adapted to the shortlisting setting.

First we consider Resistance to Clones.
In many shortlisting scenarios, for example in the context of recommender systems,
it is not always clear if alternatives should be bundled together.
For example, if we want to select a number of books to recommend,
should we include each part of a trilogy separately or bundle the whole series?
Shortlisting rules that satisfy Resistance to Clones are useful because the 
outcome of the rule is the same in both cases (if all parts of the series are 
equally popular).

\begin{Axm}[Resistance to Clones]
Adding a clone of an alternative to an election does not change the outcome, i.e.,
if $E = (C,V)$ and $E^* = (C \cup \{c^*\}, V^*)$ are
two elections with $\lvert V\rvert = \lvert V^*\rvert$ such that, for all $j \leq n$,
we have $c_i \in v_j$ if and only if $c_i \in v_j^*$
for all $c_i \in C$ and $c^* \in v_j^*$ if and only if
$c_k \in v_j$ for some $k \leq m$, then $\cR(E) = \cR(E^*)$
if $c_k \not \in \cR(E)$ and $\cR(E^*) = \cR(E) \cup \{c^*\}$ if $c_k \in \cR(E)$.
\end{Axm}

Clearly, Independence implies Resistance to Clones. Hence, \Threshold{} satisfies 
Resistance to Clones.
As cloning does not change the maximal score, the same holds for \MSThreshold.
Furthermore, cloning has no effect on gaps, hence \Lgap{} and \Fgap{} satisfy Resistance to Clones.
If follows that \Approval{} also satisfies Resistance to Clones as it is a special case of \Fgap{}.

For \First{} and \Next{} it can be helpful for an alternative to be cloned.
For example, consider an election with $\score_E = (3,2,0)$.
Then the set of \First{} winners would be
$\{c_1\}$ but after cloning $c_2$, the set of \First{} winners 
is $\{c_1,c_2,c_2'\}$ where $c_2'$ is the clone of $c_2$.
Similarly let $\score_E = (2,1,0)$. Then the set of \Next{} winners for every $k$
would be $\{c_1\}$. If we clone $c_2$, then, for all $k \geq 2$, \Next{} selects
$\{c_1,c_2,c_2'\}$ 
 where $c_2'$ is the clone of $c_2$. 

Moreover, \qNCSA{} also does not satisfy Resistance to Clones.
Consider for example $0.5$-\textit{NCSA}, assume we have $10$ voters and let $\score_E = (10,7,7)$.
Then the $0.5$-\textit{NCSA}-scores of $\{c_1\}$, $\{c_1,c_2\}$ and $\{c_1,c_2,c_3\}$
are $\nicefrac{10}{1}$, $\nicefrac{14}{\sqrt{2}} \approx 9.899$ and
$\nicefrac{18}{\sqrt{3}} \approx 10.392$ respectively. Therefore $\{c_1,c_2,c_3\}$ is
the winner set. Now, if we add a clone $c_1^*$ of $c_1$ we get the $0.5$-\textit{NCSA}-scores  
$\nicefrac{10}{1}$, $\nicefrac{20}{\sqrt{2}} \approx 14.142$,
$\nicefrac{24}{\sqrt{3}} \approx 13.856$ and
$\nicefrac{28}{\sqrt{4}} = 14$ for $\{c_1\}$, $\{c_1, c_1^*\}$, $\{c_1, c_1^*,c_2\}$
and $\{c_1, c_1^*,c_2,c_3\}$ respectively. Therefore, $\{c_1, c_1^*\}$ is winning.

Finally, 
\SPAV{} generally does not satisfy Resistance to Clones
as cloning may harm an alternative.
For example, consider $2 \vartriangleright 3 \vartriangleright \dots$
and $\score_E = (2,1,0)$. Then the set of \SPAV{} winners is
$\{c_1, c_2\}$, but if we clone $c_1$, then $c_2$ is not a winner any more.
This also shows that neither \ISPAV{} nor \TopFgapSK{} are resistant to clones (set $s= 2$ and $k=2$
for the latter).

The first monotonicity axiom we consider is Set Monotonicity. It states that if one voter additionally approves the winner set, this must not change the outcome.

\begin{Axm}[Set Monotonicity] 
For any two elections $E = (C,V)$ and $E^* = (C,V^*)$ with $V = (v_1, \dots, v_n)$ and $V^* = (v^*_1, \dots, v^*_n)$,
if there exists a $j \leq n$ such that $v_j \cap \cR(E) = \emptyset$, $v^*_j  = v_j \cup \cR(E)$
and $v^*_l = v_l$ for all $l \neq j$,
then $\cR(E^*) = \cR(E)$.
\end{Axm}

All of our rules except \First{} and \MSThreshold{} with non-constant threshold function
satisfy Set Monotonicity
\begin{itemize}
\item Let $E$ be an election with $\score(E) = (2,2,1,1,1,1)$.
Then under \First{} we have $\cR(E) = \{c_1, c_2,c_3\}$.
Now if a voter who did not approve $\{c_1, c_2,c_3\}$
before approves it, then we get
$\score(E) = (3,3,2,1,1,1)$ and hence  $\cR(E) = \{c_1, c_2\}$.
\item For \MSThreshold{} first assume
\[f(n) := \begin{cases}
           n-1 &\text{ if $n$ is odd},\\
           1   &\text{ otherwise}.
           \end{cases}
\]
Let $\score(E) = (5, 2)$. Then $c_1$ is the only winner, but after adding one approval
to $c_1$ the winner set becomes $\{c_1,c_2\}$.

Now, assume $f(n) = \alpha\cdot n$ for some $0 \leq \alpha < 1$. 
First assume $\score_E(c_i) > \alpha \max(\score(E))$ and hence $c_i \in \cR(E)$.
Then after adding one approval to all winning candidates, we have
\begin{multline*}
\score_{E^*}(c_i)=\score_E(c_i) + 1> \alpha \max(\score(E)) +1 \geq\\ \alpha(\max(\score(E)) +1) =
  \alpha\max(\score(E^*)).
\end{multline*}
This implies $c_i \in \cR(E^*)$.
On the other hand, assume $\score_E(c_i) \leq \alpha \max(\score(E))$ and
hence $c_i \not \in \cR(E)$. Then $\score_{E^*}(c_i)=\score_E(c_i) < \alpha(\max(\score(E))) \leq
\alpha\max(\score(E^*))$. It follows that $c_i \not \in \cR(E^*)$.
Therefore, \MSThreshold{} satisfies Set Monotonicity for constant $f$.
\item Clearly, adding approvals for all winners can only increase the gap between winners
and non-winners. Hence \Fgap{} and \Lgap{} satisfy Set Monotonicity.
\Approval{} is a special case of \Fgap{} and hence also satisfies Set Monotonicity.
\item For \Threshold{} clearly all winning candidates are still above the threshold in $E^*$
and all non-winning candidates remain below the threshold. Hence Set Monotonicity is satisfied.
\item \SPAV: It is easy to see that a set $\{c_1, \dots, c_i\}$ is non-tiebreaking in $E$ if and only if it is non-tiebreaking in $E^*$. Hence, \SPAV{} satisfies Set Monotonicity.
\item \Next:  Let $\{c_1, \dots, c_{\ell}\}$ be the winner set. First, let $i < \ell$.
By choice of $\ell$ we have $\score_E(c_i) \leq \sum_{j = 1}^k \score_E(c_{i+j})$.
As $i +1 \leq \ell$ we have $\score_{E^*}(c_i) = \score_E(c_i) +1 \leq
\sum_{j = 1}^k \score_E(c_{i+j}) + 1 \leq \sum_{j = 1}^k \score_{E^*}(c_{i+j})$.
On the other hand $\score_{E^*}(c_{\ell}) > \score_E(c_{\ell}) >
\sum_{j = 1}^k \score_E(c_{{\ell}+j})  = \sum_{j = 1}^k \score_{E^*}(c_{{\ell}+j})$.
It follows that $\{c_1, \dots, c_{\ell}\}$ is also the winner set under $E^*$.
\item \qNCSA: 
Let $W=\{c_1, \dots, c_k\}$ be a the largest set with maximum \qNCSA-score. 
It holds that
\begin{align*}
\score_{E^*}^{\qNCSA}(W)&=\frac{1}{\lvert W\rvert^q}\sum_{c \in W}(2\score_{E^*}(c)  - n)=\\
&=\frac{1}{\lvert W\rvert^q}\sum_{c \in W}(2(\score_E(c) + 1) - n)=\\
&= \score_{E}^{\qNCSA}(W) + \frac{2\lvert W\rvert}{\lvert W\rvert^q}\\
&= \score_{E}^{\qNCSA}(W) + 2\lvert W\rvert^{1-q}
\end{align*}
Now consider a set $W'=\{c_1, \dots, c_i\}$. Consider first $i < k$.
Then we have by the same argument as above
\[
\score_{E^*}^{\qNCSA}(W')=
 \score_{E}^{\qNCSA}(W') + 2\lvert W'\rvert^{1-q}
\]
Now, by the choice of $k$ we have $\score_E^{\qNCSA}(W') \leq \score_E^{\qNCSA}(W')$
and because $i < k$ we have $\lvert W' \rvert < \lvert W \rvert$.
It follows that $\score_{E^*}^{\qNCSA}(W') \leq \score_{E^*}^{\qNCSA}(W')$.

Now assume $k < i$. Then we have 
\begin{align*}
\score_{E^*}^{\qNCSA}(W')&=\\
&=\frac{1}{\lvert W'\rvert^q}(\sum_{j \leq k}(2(\score_E(c_j) + 1) - n) + \sum_{j > k}(2\score_E(c_j) - n)=\\
&= \score_{E}^{\qNCSA}(W') + \frac{2\lvert W\rvert}{\lvert W'\rvert^q}
\end{align*}
Again, by the choice of $k$ we have $\score_E^{\qNCSA}(W') \leq \score_E^{\qNCSA}(W')$.
Moreover, because $i > k$ we have $\lvert W' \rvert > \lvert W \rvert$
and hence 
\[\frac{2\lvert W\rvert}{\lvert W' \rvert^q} < \frac{2\lvert W\rvert}{\lvert W\rvert^q}\]
It follows that $\score_{E^*}^{\qNCSA}(W') \leq \score_{E^*}^{\qNCSA}(W')$
and hence $W$ is still the largest set with maximal \qNCSA-score.
\item Finally, consider \TopFgapSK. 
Assume first that the set of winners $W'$ under \Fgap{} is smaller than $s$. As
\Fgap{} satisfies Set Monotonicity, $W'$ remains the winner set in $E^*$.
It is still smaller than $s$ and therefore still the winner under \TopFgapSK.
Now assume that $W'$ is larger than $s$. Then, the winner set of \TopFgapSK{}
has size at least $s$. Adding one approval to the first $s$ alternatives does not create 
a new $k$ gap between them. It follows that the winner set under \Fgap{}
is still larger then $s$. This means by definition that the winner set of \TopFgapSK{}
in $E$ and $E^*$ is the winner set of \ISPAV. As \ISPAV{}
satisfies Set Monotonicity, we can conclude that \TopFgapSK{} does so as well.

\end{itemize}
Set Monotonicity is a very natural axiom for many 
applications, so the fact that \First{} does not satisfy it
makes it hard to recommend the rule in most situations.
We can strengthen this axiom as follows:
a voter that previously disapproved all winning alternatives changes her mind and 
now approves a superset of all (previously) winning alternatives; this should not
change the set of winning alternatives.
This is a useful property as it guarantees that if an additional voter enters the election, who agrees with the set of currently winning alternatives but might approve additional alternatives, then the set of winning alternatives remains the same and, in particular, does not expand.

\begin{Axm}[Superset Monotonicity]
Let $E = (C,V = (v_1, \dots, v_n))$ be an election.
If $E^* = (C,V^* = (v^*_1, \dots, v^*_n))$  is another election such that for some $j \leq n$
we have $v_j \cap \cR(E) = \emptyset$, 
$\cR(E) \subseteq v^*_j$
and $v^*_l = v_l$ for all $l \neq j$,
then $\cR(E) = \cR(E^*)$.
\end{Axm}

In contrast to Set Monotonicity, only few rules satisfy Superset Monotonicity.
Let us first show that \First, \Threshold, \MSThreshold, \Next, \Lgap{} and \SPAV{}
do not satisfy Superset Monotonicity.

\begin{itemize}
\item Clearly, Superset Monotonicity implies Set Monotonicity,
hence \First{} cannot satisfy Superset Monotonicity.
\item First, consider an election $E$ with $n= 3$ such that $\score(E) = (2,1)$.
Then $\cR(E) = \{c_1\}$ under \Threshold{} with $f = \nicefrac{n}{2}$.
Now, if one voter additionally approves $\{c_1,c_2\}$, 
then $\cR(E) = \{c_1,c_2\}$.
\item Next, consider an election $E$ such that $\score(E) = (4,2)$.
Consider \MSThreshold{} with $f(n) = \nicefrac{n}{2}$. Then
$\cR(E) = \{c_1\}$. Now, if one voter additionally approves $\{c_1,c_2\}$, 
then $\cR(E) = \{c_1,c_2\}$.
\item
For \Next, consider an election $E$ such that $\score(E) = (3,1,1)$.
Then the winner set under \Next{} is $\{c_1\}$.
Now, if a voter changes her mind and additionally approves 
all three alternatives, then all three alternatives become winners
under \Next{} (for every $k>1$).
\item 
Next, consider an election $E$ such that $\score(E) = (2,1,0)$.
For \Lgap, $\cR(E) = \{c_1\}$.
If one voter additionally approves $\{c_1,c_2\}$, 
then $\score(E) = (3,2,0)$ and $\cR(E) = \{c_1,c_2\}$.
\item 
For \SPAV, consider an election $E$ with $\score(E) = (2,1,1)$
and $2 \vartriangleright 1 \vartriangleright 3 \vartriangleright 0$. Then $\cR(E) = \{c_1\}$.
Now, if one voter additionally approves $\{c_1,c_2\}$, 
then $\cR(E) = \{c_1,c_2\}$.
\item For \textit{$0.5$-NCSA}, consider an election $E$ with $\score(E) = (90,90,67)$ and $n=98$.
Here, the winner set is $\{c_1, c_2\}$ with $\score_{E}^{\qNCSA}(\{c_1,c_2\})\approx 115.97$ and 
$\score_{E}^{\qNCSA}(\{c_1,c_2, c_3\})\approx 115.47$.
However, for $\score(E^*) = (91,91,68)$ (one voter who previously approved no one, now approves every candidate), we obtain
$\score_{E^*}^{\qNCSA}(\{c_1,c_2\})\approx 118.79$ and 
$\score_{E^*}^{\qNCSA}(\{c_1,c_2, c_3\})\approx 118.93$.
\end{itemize}

In contrast, \ISPAV{} satisfies Superset Monotonicity as any ties between winners remain.
Moreover, as the size of the gap between winners and non-winners
cannot decrease and gaps within the winner set remain, \Fgap{} satisfies
Superset Monotonicity for all $k$ (which includes \Approval). 
For this reason \TopFgapSK{} also satisfies Superset Monotonicity by an analogous argument
as for Set Monotonicity.

\iffalse
This is essentially the only case for which \SPAV{} satisfies Superset Monotonicity.
If an instance of \SPAV{} is not increasing then there are $k \vartriangleright m$
and $l\vartriangleright m$ such that $k > l$ and $k \vartriangleright l$.
We have to additionally assume that $0$ is not the most preferred option,
in which case Superset Monotonicity would be vacuously satisfied.
Then consider an election $E$ where $l$ candidates have score 
$1$ and the rest of the candidates has score $0$. By assumption, the $l$ candidates
with score $1$ are winners under \SPAV{} with $\vartriangleright$. Now, if 
a new voter approves the $l$ candidates with score $1$ and additionally $k-l$ 
other candidates, then the $k$ candidates with score $\neq 0$ are winners 
under \SPAV{} with $\vartriangleright$.}
\fi
In general, the axioms discussed in this section can be seen as axioms
about the stability of the winner set under specific changes to the election.
We observed that \First{} and, to a lesser degree, \SPAV{} and
\Next{} did not perform well in this regard. On the other hand, 
it seems that the winner set of \Fgap{} and \Approval{} are particularly stable, as they are the
only rule that satisfies all three axioms considered in this section.

\section{Clustering Algorithms as Shortlisting Methods}\label{sec:clustering}

Let us briefly discuss the relation between clustering algorithms and shortlisting methods.
The goal of shortlisting is essentially to classify 
some alternatives as most suitable based on their approval score.
The machine learning literature offers a wide variety of 
clustering algorithms that can perform such a classification.

In the following, we describe how any clustering algorithm can be translated into an approval-based variable multi-winner rule that satisfies Anonymity. For most clustering algorithms, the corresponding rule also satisfies Neutrality, Efficiency and is non-tiebreaking, and thus yields a shortlisting method.
The procedure works as follows:
Let $E=(C,V)$.
We use $\score(E)$ as input for a clustering algorithm.
This algorithm produces a partition $S_1,\dots,S_\beta$ of $\score(E)$.
The winner set is
the partition that contains the highest score, i.e., the winner set consists of those candidates whose scores are contained in the selected partition.

As this procedure is based on $\score(E)$, the resulting approval-based variable multi-winner rule is clearly anonymous.
To show that the resulting rule is a shortlisting rule, we require the following two additional assumptions:
\begin{enumerate}
\item \emph{The clustering algorithm yields the same result for any permutation of $\score(E)$.} If this is the case, the resulting rule is also neutral.
\item \emph{The algorithm outputs clusters that are non-intersecting intervals.} If this is the case, the result rule is non-tiebreaking (since clusters do not intersect). It is also efficient, as the ``winning'' cluster is an interval containing the largest score.
\end{enumerate}
These are indeed conditions that any reasonable clustering algorithm satisfies. 

As an illustration, let us consider linkage-based algorithms~\citep{shalev2014understanding}.
Linkage-based algorithms work in rounds and start with the partition of $\score(E)$ into singletons. Then, in each round, two sets (clusters) are merged until a stopping criterion is satisfied.
One important type of linkage-based algorithms are those where always the two clusters with minimum distance are merged.
Thus, such algorithms are specified by two features: a distance metric for sets (to select the next sets to be merged) and a stopping criterion.
We assume that if two or more pairs of sets have the same distance, then the pair containing the smallest element are merged.
Following  Shalev-Shwartz and Ben-David~\citep{shalev2014understanding}, we consider three distance measures:
the minimum distance between sets (Single Linkage):
\begin{align}
d_{\min}(A,B) &= \min\left\{\lvert x-y\rvert: x\in A, y\in B\right\},\label{eq:min}
\end{align}
the average distance between sets (Average Linkage)
\begin{align}
d_{\mathrm{aver}}(A,B) &= \frac{1}{\lvert A\rvert\lvert B\rvert}\sum_{x\in A, y\in B} \lvert x-y\rvert,\label{eq:aver}
\end{align}
and the maximum distance between sets (Max Linkage)
\begin{align}
d_{\max}(A,B) &= \max\left\{\lvert x-y\rvert: x\in A, y\in B\right\}.\label{eq:max}
\end{align}

These three methods can be combined with arbitrary stopping criteria; we consider two:
(A)~stopping as soon as only $\beta$~clusters remain, and (B)~stopping as soon as every pair of clusters has a distance of $\geq \alpha$.
Interestingly, two of our previously proposed methods correspond to linkage-based algorithms:
First, if we combine the minimum distance with stopping criterion~(A) for $\beta=2$,
we obtain the \Lgap{} rule.
Secondly, if we use the minimum distance and impose a distance upper-bound of $\alpha=k$ (stopping criterion~B), we obtain the \Fgap{} rule.
Thirdly, if we seek winner sets of size roughly $\nicefrac m k$ for some positive integer~$k$,
stopping criterion~(A) with $\beta=k$ is a possible choice.

We see that the literature on clustering algorithms yields a large number of shortlisting methods. The inherent disadvantage of this approach is that cluster algorithms generally treat all clusters as equally important whereas for shortlisting methods the winning set of candidates is clearly most important. This difference becomes most pronounced when a clustering algorithm produces several clusters; only the ``winning'' cluster is relevant for the resulting shortlisting method.
That being said, we identified two clustering algorithms that indeed corresponded to sensible shortlisting methods (\Fgap{} and \Lgap{}), showing that this approach can be fruitful.

\section{Experiments}\label{sec:experiments}

In numerical experiments, we want to evaluate the characteristics of the considered shortlisting rules. The Python code used to run these experiments is available \citep{martin_lackner_2020_3821983}).
We use three data sets for our experiments: two synthetic data sets (``bias model'' and ``noise model'') as well as data from a real-world shortlisting scenario, the nomination process for the Hugo awards.

\subsection{Synthetic Data}

\subsubsection*{Basic setup}
Both synthetic data sets have the same basic setup.
We assume a shortlisting scenario with 100 voters and 30 alternatives.
Each alternative $c$ has an objective quality~$q_c$, which is a real number in $[0,1]$.
For each alternative, we generate~$q_c$ from a truncated normal (Gauss) distribution with mean $0.75$ and standard deviation $0.2$, restricted to values in $[0,1]$.
This is chosen to model difficult shortlisting scenarios with several strong candidates (with an objective quality $q_c$ close to 1).
Our base assumption is that voters approve an alternative with likelihood $q_c$.
Thus, the approval score of alternatives are binomially distributed, specifically $\score_E(c)\sim B(100, q_c)$.
We then modify this assumption to study two complications for shortlisting: imperfect quality estimates (noise) and biased voters.

\subsubsection*{The noise model}
This model is controlled by a variable $\lambda\in[0,1]$.
We assume that voters do not perfectly perceive the quality of alternatives, but with increasing $\lambda$ fail to differentiate between alternatives.
Instead of our base assumption that each voter approves an alternative $c$ with likelihood $q_c$, we change this likelihood to $(1-\lambda)q_c + 0.5\lambda$.
Thus, for $\lambda=0$ this model coincides with our base assumption; for $\lambda=1$ we have complete noise, i.e., all alternatives are approved with likelihood $0.5$.
As $\lambda$ increases from 0 to 1, the amount of noise increases, or, in other words, the voters become less able to judge the quality of alternatives.

\subsubsection*{The bias model}
In this model we assume that a proportion of the voters are biased against (roughly) half of the alternatives; we call these alternatives disadvantaged.
Biased voters approve these alternatives only with likelihood $0.5\cdot q_c$, i.e., they perceive their quality as only half of their true quality.
We assign each alternative with likelihood 0.5 to the set of disadvantaged alternatives.
In addition, the alternative with the highest quality is always disadvantaged.\footnote{
We make this assumption because a bias \emph{only} against low-quality alternatives is actually helpful for the shortlisting task---this effect would distort the negative consequences of bias.}
We control the amount of bias via a variable $\gamma\in[0,1]$: a subset of voters of size $\left\lfloor 100\cdot \gamma\right\rfloor$ is biased; for the remaining voters our base assumption applies.
As in the noise model, as $\gamma$ increases from 0 to 1 the shortlisting task becomes harder as the approval scores less and less reflects the actual quality of alternatives.

\subsubsection*{Instances}
For each of the two models, we generate $1,000$ instances for each $\lambda\in\{0, 0.05, 0.1, \dots, 0.95, 1\}$, thus resulting in $20,000$ instances per model.

\subsection{The Hugo Awards Data Set}

The Hugo Awards are annual awards for works in science-fiction. Each year, awards are given in roughly 20 categories. The Hugo awards are particularly interesting for our paper as the nomination of candidates is based on voting and the submitted votes are made publicly available (this distinguishes the Hugo awards from many other literary awards with confidential nomination procedures).

The Hugo shortlisting (nomination) process works as follows. Each voter can nominate up to five candidates per category. This yields an approval-based election exactly as defined in Section~\ref{sec:formalmodel}. For each category, a shortlist of (usually) six candidates is selected.
This shortlist, however, does not necessarily consist of the six candidates with the largest approval scores. Instead, a voting rule called ``E Pluribus Hugo'' is used. This is not a shortlisting rule in our sense (Definition~\ref{def:shortlisting}), since it is not Non-tiebreaking and fails Efficiency.\footnotemark{} However, ``E Pluribus Hugo'' generally selects candidates with high approval scores and hence the actual winners are always among the top-seven candidates with the largest approval scores.
In Figure~\ref{fig:hugo-winning}, we display in which position (when sorted by approval scores) the actual winner in the second stage is found.
Note that there are three instances where a candidate in position~7 is winning. As ``E Pluribus Hugo'' always selects six candidates, this shows that either Non-tiebreaking or Efficiency is violated in these instances.

\footnotetext{We briefly describe ``E Pluribus Hugo''. This is an approval-based variable multi-winner rule based on an elimination process with two scores: approval scores and fractional approval scores. Let $N_E(c)=\{i\in N : c\in v_i\}$. Fractional approval scores are defined as $\fscore_E(c_j)=\sum_{i\in N_E(c_j)} \frac{1}{\lvert v_i\rvert}$, i.e., voters can contribute at most $1$ to the total score of all candidates.
Each round the two candidates with the lowest fractional approval scores are selected.
Out of these two, the one with the lower approval score is eliminated. This step is repeated with a reduced set of candidates (and updated fractional approval scores) until only six candidates remain. We omit details how ties are handled in this process and refer to \citet{HugoManip}, who introduced ``E Pluribus Hugo'' under the name SDV-LPE. This paper also contains a discussion
of why this rule was chosen (in reaction to strategic voting in previous years) and its merits for this specific application.}

\begin{figure}
\centering
\includegraphics[width=0.7\textwidth]{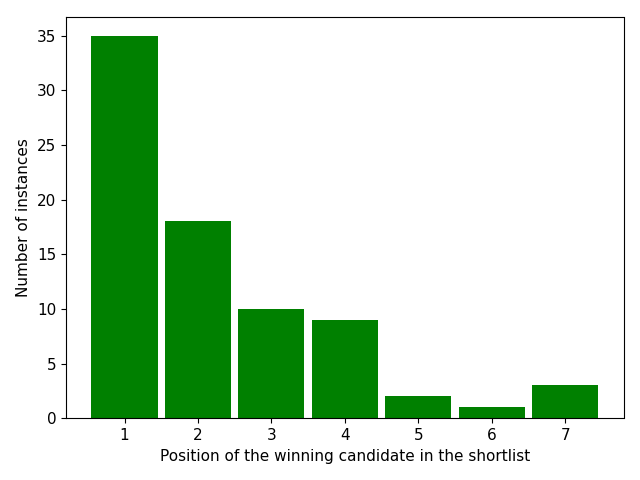}
\caption{Shortlist positions of the actual winners when sorted by approval scores.}
\label{fig:hugo-winning}
\end{figure}

Our data set is based on the years 2018--2021, comprising a total of 78 shortlisting elections.
The voting data for these years is publicly available on the Hugo website \url{https://www.thehugoawards.org/}.
For each election we recorded the actual winner in the second stage (also based on voting, but with a different, larger set of voters).  The data files are available along-side our code \citep{martin_lackner_2020_3821983}.

In a sense this is an ideal data set to test our results, as the scenario exactly matches our formal model. However, there are two caveats to be noted.
First, the true winner is always among the first seven candidates. Thus, \ispav{7} will always select a shortlist containing the true winner.
Conversely, any shortlisting rule that outputs shortlists with more than seven candidates is non-optimal on this data set. This peculiarity has to be kept in mind when interpreting our results.

Secondly, the shortlisting process of the Hugo awards has been a contentious matter with recorded attempts of organized strategic voting (this is described briefly by \citet{HugoManip}).
As a consequence, the voting results in the shortlisting stage can differ significantly from 
the results in the second stage (with a much larger set of voters). 
It is therefore reasonable to assume that this data set contains ``hard'' instances, i.e., it is difficult to find short shortlists.

\subsection{Precision and average size}

We use two metrics to evaluate shortlisting rules. To be able to speak about successful shortlisting, we assume that we know for each shortlisting instance $E_\ell$
the actual winner in the second stage, i.e., the candidate that is the winner \emph{among} shortlisted candidates; let this candidate be $c^*_\ell$.
For the synthetic data sets, $c^*_\ell$ is the candidate with the highest objective quality;
for the Hugo data set it is the candidate that actually won the Hugo award (which was selected from the shortlisted candidates).

Given a set of shortlisting instances $\{E_1,\dots,E_N\}$, we evaluate a shortlisting rule $\cR$  with respect to the following two metrics.

\begin{enumerate}
\item \emph{Precision} is the true winner ($c^*_\ell$) being contained in $\cR$'s winner sets:
\begin{align}
\frac{1}{N}\cdot \left\lvert \left\{1\leq \ell \leq N : c^*_\ell \in \cR(E_\ell) \right\}\right\rvert.
\end{align}
\item \emph{Average size} is the average size of $\cR$'s winner sets:
\begin{align}
\frac{1}{N}\sum_{\ell=1}^N \lvert \cR(E_i)\rvert.
\end{align} 

\end{enumerate}

A shortlisting is desirable if it has a high precision and small average size.
However, observe that these two metrics are difficult to reconcile. 
The easiest way to achieve high precision is to output large shortlists, and conversely, 
a small average size will likely result in a lower precision.

\subsection{Experiment 1: Increasing Noise and Bias}

Experiment~1 applies only to the two synthetic data sets. The goal is to see how different shortlisting rules deal with increasingly noisy/biased data.
We restrict our attention to six shortlisting rules, for which the results are particularly 
instructive:
\Approval{},  \Ffgap{}, \Threshold{}, \SPAV{}, \First{}, \textit{Top-$10$-First-$5$-Gap}, \Lgap{}, and  \emph{$0.5$-NCSA}.
For \SPAV{}, we use the priority order $4\rhd 5 \rhd 6 \rhd \dots$,
i.e., we use the \ispav{4} rule.
Finally, we choose \Majority{} as representative for threshold rules.
The results for \emph{Max-Score-$f$-Threshold} with $f(x)=0.5x$ were very similar to 
\Majority{} and are thus omitted.

    \begin{figure}
        \centering
        \begin{subfigure}[b]{0.5\textwidth}  
            \centering 
            \includegraphics[width=\columnwidth]{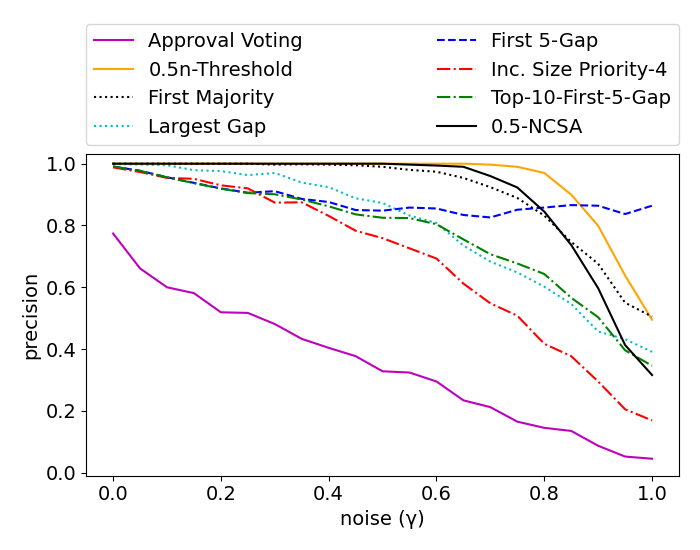}
            \caption{Inclusion of objectively best alternative} %
            \label{fig:noise-normal-inclbest}
        \end{subfigure}%
		\hfill
        \begin{subfigure}[b]{0.5\textwidth}   
            \centering 
            \includegraphics[width=\columnwidth]{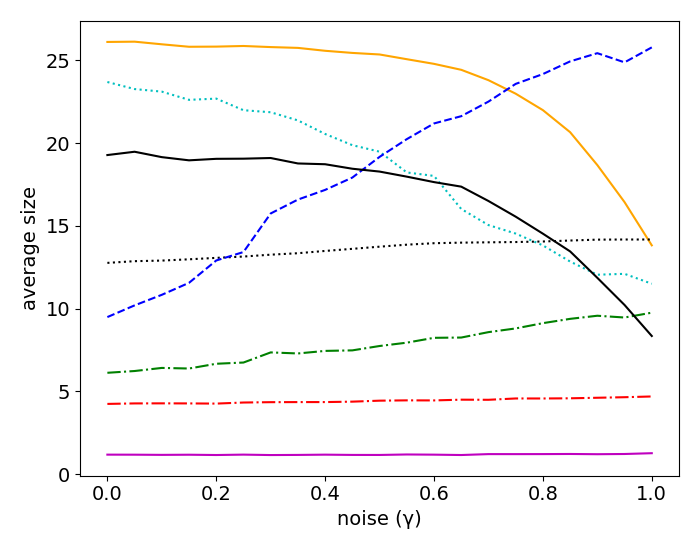}
            \caption{Average size of winner sets}    
            \label{fig:noise-normal-size}
        \end{subfigure}
   \caption{Numerical simulations for the noise model}\label{fig:noise-normal}
\end{figure}

Our comparison of shortlisting rules is visualized in Figure~\ref{fig:noise-normal} for the noise model and Figure~\ref{fig:bias-normal} for the bias model. %
Each data point in these figures (corresponding to a specific~$\lambda$) is based on $N=1000$ instances $E_1,\dots, E_{N}$.

The orthogonal nature of precision and average seen
can be seen clearly when comparing \Approval{} and  \Majority{}: \Approval{} returns rather small winner sets (as seen in Figs.~\ref{fig:noise-normal-size} and~\ref{fig:bias-normal-size}), but if~$\lambda$ %
increases,
the objectively best alternative is often not contained in the winner set (Figs.~\ref{fig:noise-normal-inclbest} and~\ref{fig:bias-normal-inclbest}). \Majority{} has large winner sets, but is likely to contain the objectively best alternative even for large $\lambda$ (up to $\lambda\approx0.8$).
If the average size of winner sets remains roughly constant (\ISPAV{}, \First{}, \Approval{}), then the precision reduces with increasing noise/bias ($\lambda$).

    \begin{figure}
        \centering
        \begin{subfigure}[b]{0.5\textwidth}  
            \centering 
            \includegraphics[width=\columnwidth]{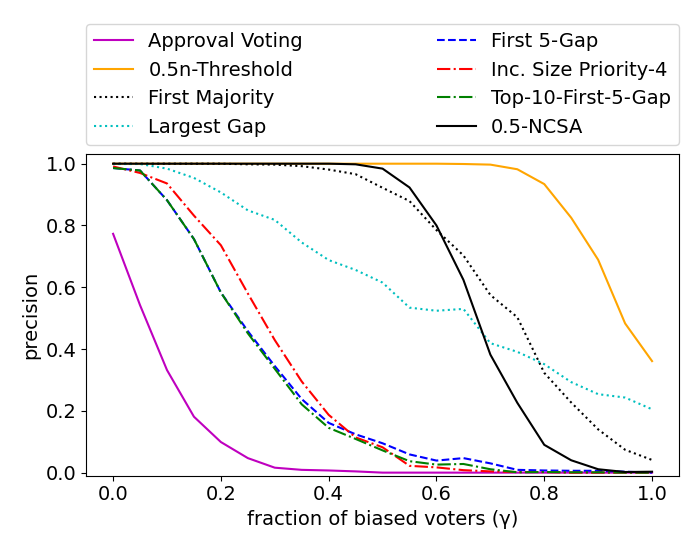}
            \caption{Inclusion of objectively best alternative} %
            \label{fig:bias-normal-inclbest}
        \end{subfigure}\hfill
        \begin{subfigure}[b]{0.5\textwidth}   
            \centering 
            \includegraphics[width=\columnwidth]{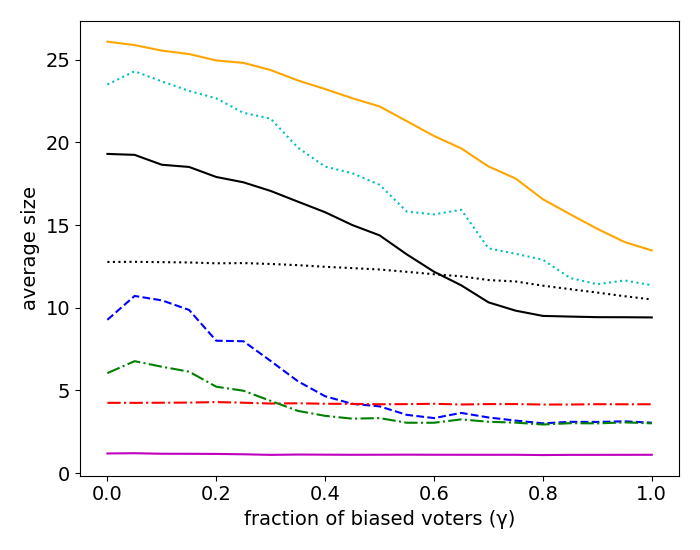}
            \caption{Average size of winner sets}    
            \label{fig:bias-normal-size}
        \end{subfigure}
   \caption{Numerical simulations for the bias model}\label{fig:bias-normal}
\end{figure}

\SPAV{} (with the considered priority order) is a noteworthy alternative to \Approval. It has an only slightly larger average size (roughly $1$ vs $4$), while having a significantly larger chance to include the objectively best alternative. As it is generally not necessary to have extremely small winner sets in shortlisting processes, we view \SPAV{} (with a sensibly chosen priority order) as superior to \Approval.

Considering the noise model (Fig.~\ref{fig:noise-normal}), 
we see a very interesting property of \Ffgap: it is the only rule where the size of winner sets significantly adjusts to increasing noise. If $\lambda$ increases, the differences between the approval scores vanishes and thus fewer $5$-gaps exist. As a consequence, the winner sets increase in size. This is a highly desirable behavior, as it allows \Ffgap{} to maintain a high likelihood of containing the objectively best alternative without producing very large shortlists for low-noise instances.

Two other rules also show this behavior: \textit{Top-$10$-First-$5$-Gap} and 
\First{}, albeit both to only a small degree.
\textit{Top-$10$-First-$5$-Gap} achieves the same precision as  \Ffgap{} until $\gamma$ reaches $\approx 0.5$ after which its precision deteriorates. On the other hand, note that \textit{Top-$10$-First-$5$-Gap} has a considerably smaller average size.
For \Lgap{}, \Majority{}, and \qNCSA{}, we see the opposite effect: winner sets are large for low noise but decrease with increasing $\lambda$. This is not a sensible behavior; note that \First{} achieves better precision with much smaller average size.

For the bias model, we do not observe any shortlisting rule that reacts to an increase
in bias with a larger average size.

To sum up, our experiments show the behavior of shortlisting rules with accurate and inaccurate voters, and the trade-off between large and small winner set sizes. In these experiments, we see two shortlisting rules with particularly favorable characteristics:
\begin{enumerate}
\item \SPAV{} produces small winner sets with good precision. Thus, it shows a certain robustness to a noisy selection process, as is desirable in shortlisting settings.
\item \Fgap{} manages to adapt in high-noise settings by increasing the winner set size, the only rule with this distinct feature. This makes it particularly recommendable in settings with unclear outcomes (few or many best alternatives), where a flexible shortlisting method is required. As we will see in the next experiment, however, \Fgap{} on its own can be insufficient, which leads us to recommending the related \TopFgapSK{} rule instead.
\end{enumerate}

\subsection{Experiment 2: Tradeoffs between Precision and Size}\label{sec:exp2}

In this second experiment, we want to study the tradeoff between precision and size in more depth and for many more shortlisting rules.
Here, we put particular emphasis on the Hugo data set (but also consider both synthetic sets).
To this end, we represent shortlisting rules as points in a two-dimensional plane with average size as x-axis and precision as y-axis. 
Figure~\ref{fig:exp2:hugo} shows these results for the Hugo data set (points are averaged over 78 instances), Figure~\ref{fig:exp2:noise} shows these results for the noise data set (no noise to moderate noise, i.e., $\lambda\in[0,0.5]$, yielding 10,000 instances),
and Figure~\ref{fig:exp2:bias} for the bias model (also for $\lambda\in[0,0.5]$, 10,000 instances).

These plots can be understood as follows. Ideal shortlisting rules lie in the top left corner (high precision, low average size). As this is generally unachievable, we have to choose a
compromise between the two metrics. 
The gray area shows the space in which such a compromise has to be found (when
choosing from shortlisting rules that are studied in this paper). 

We will now explain the gray area in more detail:
For Experiment~2, we consider all shortlisting rules defined in Section~\ref{sec:rules} with the following parameters.
For $\alpha\in\{0,0.01, 0.02, \dots, 1\}$, we consider
\begin{itemize}
\item \Next{} for $k\in\{2,3\}$,
\item \Threshold{} and \emph{Max-Score-$f$-Threshold} with $f(n)=\lfloor\alpha \cdot n\rfloor$,
\item \qNCSA{} with $q=\alpha$,
\item \Fgap{} with $k=\lfloor\alpha \cdot n \rfloor$ and with $k=\lfloor\alpha \cdot \max{\score(E)}\rfloor$,
\item \ISPAV{} with priority orders of the form $s\rhd s+1\rhd \dots$ (\ispav{s}) for $2\leq s\leq m$,
\item \TopFgapSK{} with $2\leq s\leq m$ and $k \in \{\lfloor\alpha \cdot n \rfloor, \lfloor\alpha \cdot \max{\score(E)}\rfloor\}$.
\end{itemize}
Each shortlisting rule yields a point in this two-dimensional space.
Shortlisting rules with one parameter are displayed as lines.
We can compute a Pareto frontier consisting of all points that do not
have another point above and to the left of it.
The boundary of the gray area shows this Pareto frontier.
Consequently, voting rules close to this frontier represent a more beneficial tradeoff between precision and average size.

\subsubsection*{Results for the Hugo data set}

When looking at Figure~\ref{fig:exp2:hugo}, we see as expected that \ispav{7} achieves a precision of $1$ and an average size slightly above 7 (due to ties).
We furthermore see that \ispav{4}, \ispav{5}, and \ispav{6} are all very close to the 
Pareto frontier.
This raises the question whether \ISPAV{} is an ideal choice for this data set.
While this class is a good choice, it can be improved by \TopFgapSK.
In Table~\ref{tab:isp-comp}, we exemplarily show the precision and average size values for \ispav{6}, and \ispav{7} alongside 
shortlisting rules that achieve a smaller average size with the same (or better) precision.
This table gives an indication how to use \TopFgapSK in a real-world shortlisting task:
First, choose a sensible maximum size of a shortlist; in the case of Hugo awards this was chosen to be six (and was five prior to 2017).
Then, identify a bound that constitutes a significant gap; this bound can be chosen conservatively.
In the Hugo data set, a sensible choice appears to be $30\%$ of voters.
That is, 
if we encounter a gap (in the sense of \Fgap{}) in $\sc(E)$ of more than $0.3n$, we cut the shortlist at this point if this leads to a shorter shortlist.

\begin{figure}
\includegraphics[width=\columnwidth]{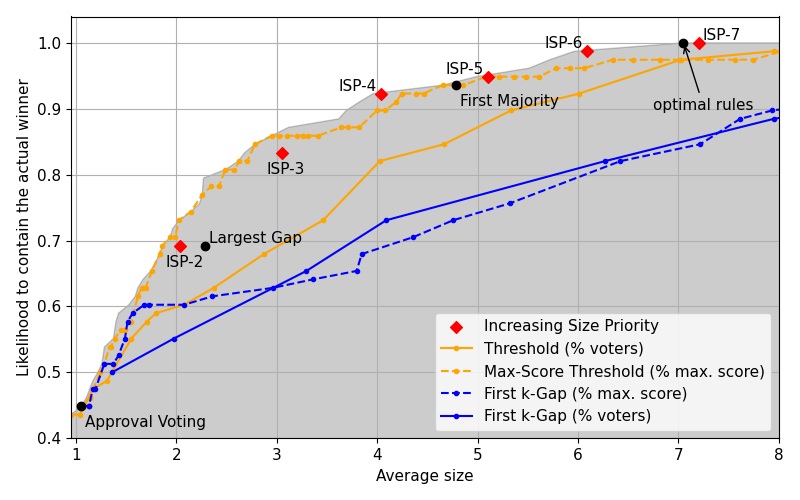}
\caption{Results for the Hugo data set (Experiment 2).}\label{fig:exp2:hugo}
\end{figure}

\begin{table}
\centering
\begin{tabular}{l| cc}
 \textbf{shortlisting rule               }& \textbf{average size }& \textbf{precision}\\
\hline
  \ispav{7}                      &       7.205  &      1.000\\
\hline
  \TopsFgapSK{7}{\lfloor\alpha \cdot n \rfloor} for $\alpha\in[0.31,0.39]$  &        7.128 &        1.000\\
  \TopsFgapSK{7}{\lfloor\alpha \cdot \max{\score(E)}\rfloor} for $\alpha\in[0.70,0.72]$  &       
        7.128 &        1.000\\
  \TopsFgapSK{7}{\lfloor\alpha \cdot n \rfloor} for $\alpha\in[0.27,0.30]$ &        7.051 &        1.000\\
  \TopsFgapSK{7}{\lfloor\alpha \cdot \max{\score(E)}\rfloor} for $\alpha=0.69$   &        7.051 &        1.000\\        
  &&\\
    \ispav{6}               &         6.090    &    0.987\\
\hline
  \TopsFgapSK{6}{\lfloor\alpha \cdot n \rfloor} for $\alpha\in[0.31,0.39]$         &       6.026 &        0.987\\
  \TopsFgapSK{6}{\lfloor\alpha \cdot \max{\score(E)}\rfloor} for $\alpha\in[0.70,0.72]$       &        6.026 &        0.987\\
  \TopsFgapSK{6}{\lfloor\alpha \cdot n \rfloor} for $\alpha\in[0.27,0.30]$         &        5.962 &        0.987\\
  \TopsFgapSK{6}{\lfloor\alpha \cdot \max{\score(E)}\rfloor} for $\alpha=0.69$     &        5.962 &        0.987
\end{tabular}
\caption{Shortlisting rules that are superior to \ispav{6} and \ispav{7} in the Hugo data set.
}\label{tab:isp-comp}
\end{table}

Let us now consider other shortlisting rules.
We see that \MSThreshold{} closely traces the Pareto frontier and thus is a very good choice 
for selecting a compromise between precision and average size. 
\Threshold{} and \Fgap{} are less convincing.
\qNCSA{} performs even worse, as very often candidates have approval scores of less than $0.5n$.
Therefore \qNCSA{} selects mostly empty sets and is thus not visible in Figure~\ref{fig:exp2:hugo}  (cf.\ Observation~\ref{obs:qncsa}).
A notable unparameterized rule is \First{}, which is very close to the Pareto frontier.

To sum up our results for the Hugo data set, we identify the following shortlisting rules
as particularly suitable.
\TopsFgapSK{7}{\lfloor\alpha \cdot n \rfloor} for $\alpha\in[0.27,0.30]$
and \TopsFgapSK{7}{\lfloor0.69 \cdot \max{\score(E)}\rfloor} achieve a precision of~1
with the smallest average size ($7.051$); in Figure~\ref{fig:exp2:hugo} these rules correspond to the point labeled ``optimal rules''. 
In general, \ISPAV{} and \MSThreshold{} achieve a very good compromises between precision
and average size.

\subsubsection*{Results for the noise and bias models}

\begin{figure}
\includegraphics[width=\columnwidth]{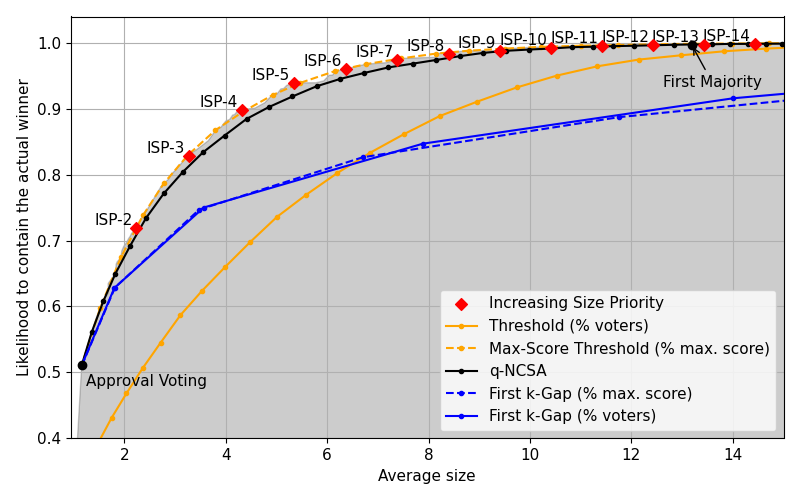}
\caption{Results for the noise model (Experiment 2).}\label{fig:exp2:noise}
\end{figure}

Figure~\ref{fig:exp2:noise} shows the results for the noise model. We see that also here \ISPAV{} and \MSThreshold{} are very close to the Pareto frontier. The same holds for \First{}.
A major difference to the Hugo data set is the performance of \qNCSA.
As candidates generally have approval scores of more than $0.5n$, \qNCSA works as intended with points close to the Pareto frontier.
As before, \Threshold{} and \Fgap{} are less convincing.

\begin{figure}
\includegraphics[width=\columnwidth]{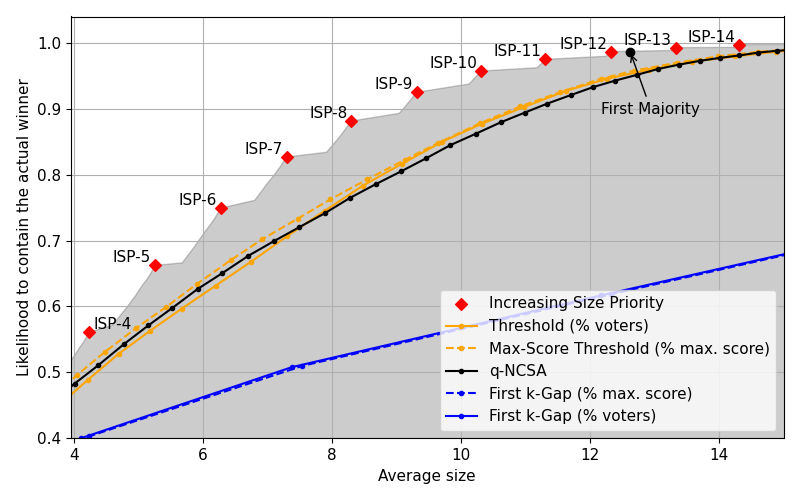}
\caption{Results for the bias model (Experiment 2).}\label{fig:exp2:bias}
\end{figure}

The bias model is a scenario, where some high-quality candidates receive too few approvals.
In Figure~\ref{fig:exp2:bias}, we see that this is a tough problem. The only recommendable
shortlisting rules are \ISPAV{} rules. By simply shortlisting the top-$k$ candidates, there
is a certain chance to also shortlist high-quality but disadvantaged candidates.
We remark that the Pareto frontier between \ispav-points is due to \TopFgapSK{} rules.

\section{Discussion}\label{sec:concl}

Based on our analysis, we recommend three shortlisting methods: \SPAV{}, \TopFgapSK{}, and \Threshold{}.
Let us discuss their advantages and disadvantages:

\begin{itemize}
\item \SPAV{}, in particular \ISPAV{}, is recommendable if the size of the winner set is of particular importance, e.g., in highly structured shortlisting processes such as the nomination for awards. 
\ISPAV{} exhibits good axiomatic properties (cf.\ Table~\ref{tab:Results}) as well as a very solid behavior in our numerical experiments.
In particular for the bias data set, where a (unknown) subset of candidates is discriminated against, \ISPAV{} appears to be the best choice. By selecting $k$~candidates with the highest approval scores (or more in case of ties), the differences in approval scores within the selected group are ignored and thus disadvantaged, high-quality candidates have a better chance to be chosen.
On the other hand, \ISPAV{} makes limited use of the available approval preferences and thus can be seen as a good choice mostly in settings with limited trust in voters' accuracy.
When voters are expected to have good estimates of the candidates' qualities, the following two shortlisting rules are better suited.

\item Our axiomatic analysis reveals \Fgap{} as a particularly strong rule in that it is the minimal rule satisfying $\ell$-Stability. 
Furthermore, it is the only rule that adapts to increasing noise in our simulations.
However, we have seen in Experiment~2 (Section~\ref{sec:exp2}) that \Fgap{} is prone to choosing winner sets that are larger than necessary.
Thus, we recommend to use \TopFgapSK{} instead. 
\TopFgapSK{} shares most axiomatic properties with \Fgap{} (cf.\ Table~\ref{tab:Results})
except $\ell$-Stability and Resistance to Clones.
Another advantage of \TopFgapSK{} is that the parameter~$k$ is difficult to 
choose for \Fgap{}, whereas it is very reasonable to conservatively pick a large $k$-value for \TopFgapSK{}. Choosing~$k$ too large simply diminishes the differences between \TopFgapSK{} and \ispav{k}.

\item Finally, 
Theorem~\ref{Class:Threshold-thm} shows that \Threshold{} rules are the only rules satisfying the Independence axiom. Therefore, if the selection of alternatives should be independent from each other, then clearly a \Threshold{} rule should be chosen.
For example, the inclusion in the Baseball Hall of Fame should depend on the quality of a player and not on the quality of the other candidates.
In our experiments, we have seen that the related class of \MSThreshold{} rules has advantages over \Threshold{} rules. The difference between these two classes, however, is only relevant if the maximum score of candidates differs between elections for reasons unrelated to the candidates' quality. This was the case, e.g., in the Hugo data set, where the relative maximum 
approval score varied significantly between award categories.
\end{itemize}

\medskip
These recommendations are applicable to most shortlisting scenarios.
There are, however, possible variations of our shortlisting framework
that require further analysis in the future.
For example, while strategyproofness is usually not important in election with
independent experts, there are some shortlisting applications with a more open electorate
where this may become an issue~\citep{HugoManip,bredereck2018coalitional}.
We have not considered strategic voting in this paper and assume that this viewpoint will
give rise to different recommendations.
Moreover, it may be worth investigating whether using ordinal preferences (rankings) instead of
approval ballots can increase the quality of the shortlisting process (shortlisting rules for ordinal preferences can be found, e.g., in the works of \citet{elkind2017multiwinner,ijcai/mw-condorcet,Faliszewski2017MultiwinnerVoting,elkind2017properties}).
In general, the class of variable multi-winner rules (and social dichotomy functions) deserves further attention
as many fundamental questions (concerning proportionality, axiomatic classifications, algorithms, etc.) are still unexplored.

\section*{Acknowledgments}
This work was supported by the Austrian Science Fund (FWF): P31890 and J4581

\clearpage
\appendix
\section{Proof of Claim~\ref{claim1} in Proposition~\ref{prop:qNCSA}}
\textbf{Claim~\ref{claim1}\ \ }
Let $E$ be an election. Then if $\score_E(c_i) = \score_E(c_{i+1})$
and \[\score_E^{\qNCSA}(\{c_1, \dots, c_{i-1}\}) \leq \score_E^{\qNCSA}(\{c_1, \dots, c_{i}\})\]
then also 
\[\score_E^{\qNCSA}(\{c_1, \dots, c_{i}\}) \leq \score_E^{\qNCSA}(\{c_1, \dots, c_{i+1}\}).\]

\begin{proof}
Assume $\score_E(c_i) = \score_E(c_{i+1})$
and $\score_E^{\qNCSA}(\{c_1, \dots, c_{i-1}\}) \leq \score_E^{\qNCSA}(\{c_1, \dots, c_{i}\})$.
To enhance readability, we write $x$ for $\sum_{k=1}^{i-1}(2\score_E(c_k) - n)$ and $y$ for
$2\score_E(c_i) - n$. Then, we can write 
\[\score_E^{\qNCSA}(\{c_1, \dots, c_{i-1}\}) \leq \score_E^{\qNCSA}(\{c_1, \dots, c_{i}\})\]
as 
\[\frac{x}{\lvert \{c_1, \dots,c_{i-1}\}\rvert^q} =\frac{x}{(i-1)^q} \leq
\frac{x+y}{i^q} = \frac{x+y}{\lvert \{c_1, \dots,c_{i}\}\rvert^q}\]
Now, defining $z$ as $i^q - (i-1)^q$, we can rewrite this as 
\[\frac{x}{(i-1)^q} \leq \frac{x+y}{(i-1)^q + (i^q - (i-1)^q)} =
\frac{x+y}{(i-1)^q + z}\]
Then we can do the following computation:
\begin{align*}
\frac{x}{(i-1)^q} &\leq \frac{x+y}{(i-1)^q + z}\\
x((i-1)^q+z) &\leq (x+y)(i-1)^q\\
x(i-1)^q+xz &\leq x(i-1)^q + y(i-1)^q\\
xz &\leq y(i-1)^q\\
xz +yz&\leq y(i-1)^q + yz\\
z(x +y)&\leq y((i-1)^q + z)\\
z\frac{x+y}{(i-1)^q +z} &\leq y\\
x+y +z\frac{x+y}{(i-1)^q +z} &\leq x+y+y\\
\frac{(x+y)((i-1)^q+z)}{(i-1)^q+z} +z\frac{x+y}{(i-1)^q +z} &\leq x+2y\\
\frac{x+y}{(i-1)^q+z}((i-1)^q+2z) &\leq x+2y\\
\frac{x+y}{(i-1)^q+z} &\leq \frac{x+2y}{((i-1)^q+2z)}\\
\end{align*}
Now replacing $x,y,z$ again by their respective definition we get
for the left-hand side:
\begin{align*}
\frac{\sum_{k=1}^{i-1}(2\score_E(c_k) - n) + (2\score_E(c_i) - n)}{(i-1)^q +(i^q - (i-1)^q)} &=
\frac{\sum_{k=1}^{i}(2\score_E(c_k) - n)}{i^q} \\&= \score_E^{\qNCSA}(\{c_1, \dots, c_{i}\})
\end{align*}
Observe that, by definition, we have  $y =2\score_E(c_{i}) - n = 2\score_E(c_{i+1}) - n$.
Therefore, we can write the right-hand side as
\[\frac{\sum_{k=1}^{i-1}(2\score_E(c_k) - n) + (2\score_E(c_i) - n) +
(2\score_E(c_{i+1}) - n)}{(i-1)^q +(i^q - (i-1)^q) + z} =
\frac{\sum_{k=1}^{i+1}(2\score_E(c_k) - n)}{i^q + z} 
\]
Now, we claim that because $0 \leq q \leq 1$ we have 
\[z = i^q - (i-1)^q \geq (i+1)^q - i^q.\] 
We observe that the both sides of the equation equal the change of the function
$x^q$ in an interval of one. Because the derivative of $f(x)=x^q$ for $0 \leq q \leq 1$
is monotone declining, we can bound this change using the slope of $f(x)$
in either the starting or end point of the interval as follows
\[1 \cdot f'(x+1) \leq f(x+1) - f(x) \leq 1 \cdot f'(x).\] 
Therefore, we have 
\[i^q - (i-1)^q \geq f'(i) \geq (i+1)^q - i^q\]
It follows that 
\[
\frac{\sum_{k=1}^{i+1}(2\score_E(c_k) - n)}{i^q + z} \leq 
\frac{\sum_{k=1}^{i+1}(2\score_E(c_k) - n)}{i^q + (i+1)^q - i^q} =
\score_E^{\qNCSA}(\{c_1, \dots, c_{i+1}\}).\]
All together we have shown
\[\score_E^{\qNCSA}(\{c_1, \dots, c_{i}\}) \leq \score_E^{\qNCSA}(\{c_1, \dots, c_{i+1}\}).\]
This concludes the proof.

\end{proof}

\end{document}